\documentclass{article}

\usepackage[english]{babel}
\usepackage{balance} % for balancing columns on the final page
\usepackage[utf8]{inputenc}
\usepackage{amsmath}
\usepackage{amsthm}
\usepackage[mathcal]{euscript}
\usepackage[T1]{fontenc}
\usepackage{url}            % simple URL typesetting
\usepackage{booktabs}       % professional-quality tables
\usepackage{nicefrac}       % compact symbols for 1/2, etc.
\usepackage{microtype}      % microtypography
\usepackage{tikz}
\usepackage[inline]{enumitem}
\usepackage{multirow}
\usepackage{diagbox}
\usepackage{mathtools}
\usepackage[letterpaper,top=2cm,bottom=2cm,left=3cm,right=3cm,marginparwidth=1.75cm]{geometry}

\usepackage{graphicx}
\usepackage[colorlinks=true, allcolors=blue]{hyperref}

\title{Facility Location Problem with Aleatory Agents}

\author{
    Gennaro Auricchio\thanks{Department of Mathematics, University of Padua, Padua, Italy. Email: gennaro.auricchio@unipd.it} \and
    Jie Zhang\thanks{Department of Computer Science, University of Bath, Bath, UK. Email: jz2558@bath.ac.uk} 
}

\date{}

\newtheorem{theorem}{Theorem}

\newtheorem{corollary}{Corollary}

\newtheorem{remark}{Remark} 

\newtheorem{problem}{Problem}
\newtheorem{example}{Example}
\newtheorem{definition}{Definition}
\newtheorem{mechanism}{Mechanism}

\def \Dnn{\Delta_{n_r,n_u}}
\def \PP{\mathcal{P}}
\def \RR{\mathcal{R}}
\def \RRn{\mathcal{R}(n_r,n_u)}
\def \EE{\mathbb{E}}
\def \erre{\mathbb{R}}
\def \enne{\mathbb{N}}
\def \ESC{\mathcal{ESC}}
\def \II{\mathbb{I}}

\def \Flam{F_{\lambda,\mu,\vec x}}
\def \fg{f_{\gamma}}
\def \PQM{\texttt{PQM}}
\def \SAR{\texttt{SAR}}

\usepackage{algorithmic}
\usepackage[inline]{enumitem}
\usepackage{subcaption}
\usepackage{amsthm}
\usepackage{amsfonts}
\usepackage{amsmath}

\DeclareMathOperator*{\argmin}{arg\,min}
\DeclarePairedDelimiter\ceil{\lceil}{\rceil}
\DeclarePairedDelimiter\floor{\lfloor}{\rfloor}

\begin{document}
\maketitle

%------
% Insert your abstract.
%------
\begin{abstract}
In this paper, we introduce and study the Facility Location Problem with Aleatory Agents (FLPAA), where the facility can accommodate a number of agents, namely $n$, which is larger than the number of agents reporting their preferences, namely $n_r$.
The spare capacity is used by $n_u := n - n_r$ aleatory agents which are assumed to be samples of a probability distribution $\mu$. 
The goal of FLPAA is to find a location $y$ that minimizes the \textit{ex-ante social cost}, which is defined as the expected cost of the $n_u$ agents sampled from $\mu$ plus the classic social cost incurred by the agents reporting their position. 
We show that there exists a discrete set that always contains at least one optimal solution.
We then investigate the mechanism design aspects of the FLPAA under the assumption that the mechanism designer lacks knowledge of the distribution $\mu$ but can query $k$ quantiles of $\mu$.
We explore the trade-off between acquiring more insights into the probability distribution and designing a better-performing mechanism, which we describe through the strong approximation ratio (SAR).
The SAR of a mechanism measures the highest ratio between the cost of the mechanisms and the cost of the optimal solution on the worst-case input $\vec x$ and worst-case distribution $\mu$, offering a stringent metric for efficiency loss that does not depend on $\mu$.
To better exemplify the challenges of our framework, we divide our study into four different information settings: \begin{enumerate*}[label=(\roman*)]
    \item the \textit{zero information case}, in which the mechanism designer has access to no quantiles,
    \item the \textit{median information case}, in which the mechanism designer has only access to the median of $\mu$,
    \item the \textit{$n_u$-quantile information case}, in which the mechanism designer has access to $n_u$ quantiles of its choice, and  
    \item the \textit{$k$-quantile information case}, in which the mechanism designer has access to $k<n_u$ quantiles of its choice.
\end{enumerate*}
For all these frameworks, we propose a mechanism that is either optimal or achieves a small constant SAR and pair it with a lower bound on the SAR.
In most cases, the lower bound matches the upper bound, proving that our mechanisms are tight, that is no truthful mechanism can achieve a lower SAR.
Lastly, we extend the FLPAA to include cases in which we must locate two facilities.
\end{abstract}

\maketitle

\vspace{1cm}

{\textbf{ Keywords.}} Facility Location Problem, Mechanism Design, Game Theory, Social Choice Theory\\

\section{Introduction}
The Facility Location Problem (FLP) is a classic problem in combinatorial optimization whose objective is to determine the optimal placement of facilities to minimize transportation costs associated with servicing customers \cite{hochbaum1982heuristics}.
Ever since its introduction, the FLP has become a key subproblem in several social choice-related topics, such as disaster relief \cite{doi:10.1080/13675560701561789}, supply chain management \cite{MELO2009401}, healthcare \cite{ahmadi2017survey}, clustering \cite{hastie2009elements}, and public facilities accessibility \cite{barda1990multicriteria}.
Due to its practical significance, various formulations of the facility location problems have garnered considerable attention across diverse fields, including operations research, theoretical computer science, economics, and computational game theory \cite{celik2020comparative}.
In economics and computational game theory, the study of the FLP has a distinct perspective. 
Instead of finding the best algorithm to compute a solution, the interest is focused on studying how to define a routine that elicits the position of a facility from the information reported by strategic customers, or \textit{agents}.
This research area is also known as Algorithmic Mechanism Design.
Since every agent needs to access a facility, they will misrepresent their information if this influences the routine to place the facilities in a position that they prefer, adding a novel layer of complexity to the FLP.
Indeed, optimizing a shared objective solely relying on the reported preferences of agents can result in undesirable manipulation fueled by the self-interested behaviour of the agents.
Hence, a crucial characteristic that a mechanism must possess is \textit{truthfulness}, which ensures that no agent can benefit by misrepresenting their private information.
Committing to truthful routines, however, leads to suboptimal solutions, thus resulting in an efficiency loss.
The standard value to quantify the trade-off between achieving the optimal objective and implementing a truthful mechanism is the \textit{approximation ratio}, which is the worst-case ratio between the objective achieved by the mechanism and the optimal objective attainable, \cite{nisan1999algorithmic}.
Since a higher approximation ratio indicates a greater deviation from the optimal solution, one of the main challenges in Algorithmic Mechanism Design is to define truthful routines whose approximation ratio is small. 
In the classic FLP problem, the facility location is decided only based on the agents' reports.
In some cases, this assumption is limiting.
For example, hospitals, bus stops, or other public facilities are open to everyone, not only to the agents engaging with the eliciting procedure.
Therefore, optimizing solely on the reports of interested agents without considering external participants does not necessarily locate the facility at the best possible place.
Moreover, the facility might be able to accommodate a large number of agents, therefore gathering an equal number of reports might be impossible or prohibitive.
Consider for example the case in which a chain of coffee shops wants to open a new branch on a street.
Some of the people living on the street are already a customer of said chain of coffee shops, thus they will engage in the process of eliciting the position of the new branch by reporting their preferences on where the new shop should be placed.
However, the number of agents that have already been customers of the agency is usually much smaller than the total number of agents the shop can serve.
In this case, it is necessary to mind that part of the population living on the street will use the facility after its opening.
Denoted with $\mu$ the probability measure describing the population on the street, the mechanism designer has to define a routine that minimizes the cost of the agents reporting their position while considering that new customers drawn from  $\mu$ will be accessing the facility as well.
To describe this scenario, we propose the Facility Location Problem with Aleatory Agents (FLPAA) and study its Algorithmic and Mechanism Design aspects.

\paragraph*{Our Contribution and Technique Overview}

In this paper, we develop a novel framework for the Facility Location Problem (FLP) where the number of agents the facility can accommodate, namely $n$, is larger than the number of agents reporting their preferences, namely $n_r$.
We assume that the spare capacity of the facility, namely $n_u:=n-n_r$, is used by external agents modelled as independent and identically distributed  (i.i.d.) samples of a probability distribution $\mu\in\PP(\erre)$, where $\PP(\erre)$ denotes the set of probability measures over $\erre$.
To evaluate the quality of a facility location $y$, we introduce the \textit{ex-ante social cost}, which measures the expected cost of $n_u$ i.i.d. agents distributed according to $\mu$ added to the classic social cost of the agents reporting their position, \textit{i.e.}
\begin{equation}
    \label{eq:introESC}
    \ESC(\vec x, y;\mu):=\sum_{i=1}^{n_r}|x_i-y|+\sum_{j=1}^{n_u}\EE_{X_j\sim \mu}[|X_j-y|]=\sum_{i=1}^{n_r}|x_i-y|+n_u\EE_{X\sim \mu}[|X-y|].
\end{equation}
Given $n\in\enne$, $\vec x\in\erre^{n_r}$ where $n_r\le n$ and a probability distribution $\mu$, the Facility Location Problem with Aleatory Agents (FLPAA) consists in finding $y\in\erre$ that minimizes the objective in \eqref{eq:introESC}.
We focus on the scenario where $n$ is odd throughout the paper. The case in which $n$ is even can be addressed similarly.

First, we fully characterize the set of optimal solutions of the FLPAA for any given $\mu$ and $\vec x$. 
We recall that $f\in\erre$ is the $q$-th quantile of $\mu$ if it holds $F_\mu(f)=q$, where $F_\mu$ is the cumulative distribution function (c.d.f.) of $\mu$.
In what follows, we use $\vec f\in\erre^{k}$ to denote the vector containing the quantiles associated with the values $\vec q\in[0,1]^k$, so that $f_j$ is the $q_j$-th quantile of $\mu$.
Notice that, given $\mu$, the vector $\vec q$ uniquely identifies $\vec f$, and \textit{vice-versa}.
We show that to retrieve the optimal solution it is not necessary to have access to a full description of $\mu$, but rather to $n_u$ carefully chosen quantiles of the distribution.
Given $F_\mu$ the c.d.f. associated with $\mu$, let $F_\mu^{[-1]}$ denote its pseudo-inverse function; then the optimal solution of the FLPAA belongs to the set $\mathcal{X}\cup\mathcal{F}_{n_u}$, where
\begin{enumerate*}[label=(\roman*)]
    \item $\mathcal{X}=\{x_i\}_{i\in[n_r]}$ is the set containing all the positions reported by the agents; and
    \item \label{point2} $\mathcal{F}_{n_u}=\{f_k\}$ is the set containing $n_u$ quantiles of $\mu$. The quantiles we need to consider depend on whether $n$ is even or odd. If $n$ is odd, we set $f_k=F_\mu^{[-1]}(\frac{2k-1}{2n_u})$. If $n$ is even, we set $f_k=F_\mu^{[-1]}(\frac{k}{n_u})$ for $k=1,\dots,n_u$.
\end{enumerate*}
This characterization holds for every couple of $n_r$ and $n_u$. However, in specific cases, it is possible to reduce the number of quantiles needed.
For this reason, we introduce the set of relevant quantiles of $\mu$ as 
\begin{equation}
\label{intro:def:rrn}
    \RRn:=\Big\{j\;\text{s.t.}\;\exists \vec x\in\erre^{n_r},\mu\in\PP(\erre);\; F_\mu^{[-1]}\Big(\frac{2j-1}{2n_u}\Big)\in\argmin_{y\in\erre}\;\ESC(\vec x, y;\mu)\Big\}.
\end{equation}
We then study the mechanism design aspects of the FLPAA.
We consider the case in which the mechanism designer does not know the distribution $\mu$, but they can query $k\in\enne$ quantiles of $\mu$.
This allows us to outline the trade-off between gathering more insight into the probability distribution $\mu$ and defining a better-performing truthful mechanism.
{\color{black} Please notice that, while the first set of $n_r$ agents report their preferences to a mechanism, the remaining $n_u$ agents do not submit their preferences to the mechanism as they become clients after the facility has opened.
Therefore, the mechanisms need to be truthful only with respects to the reports of $n_r$ agents, while the Social Cost has to be computed with respect to all the $n_r+n_u$ agents.}
The challenge posed by our framework is to determine which quantiles to query in order to define a mechanism whose cost is as close as possible to the optimal one, regardless of the agents' report $x$ and the distribution $\mu$.
For this reason, we introduce the notion of strong approximation ratio (SAR) which evaluates the mechanism on the worst-case input $\vec x$ and the worst possible distribution $\mu$, that is
\[
    \SAR(M):=\sup_{\mu\in\PP(\erre)}\sup_{\vec x\in\erre^{n_r}}\frac{\ESC_{M}(\vec x; \mu)}{\ESC_{opt}(\vec x;\mu)}
\]
where $M$ is a truthful mechanism, $\ESC_{M}(\vec x;\mu):=\ESC(\vec x, M(\vec x);\mu)$ denotes the ex-ante social cost achieved by the mechanism, and $\ESC_{opt}(\vec x;\mu):=\min_{y\in\erre}\ESC(\vec x,y;\mu)$ is the optimal ex-ante social cost.
Notice that the SAR is a stricter metric than the classic approximation ratio, as it provides efficiency guarantee bounds that do not depend on $\mu$.
Throughout our study, we focus our attention on the set of Phantom Quantile Mechanisms (PQM).
Given a vector $\vec w\in[0,1]^{n_u}$, the PQM associated with $\vec w$, namely $\PQM_{\vec w}$, is defined as $\PQM_{\vec w}(\vec x):=med(\vec x,\vec f)=med(\vec x,(F_\mu^{[-1]}(w_1), F_{\mu}^{[-1]}(w_2),\dots, F_{\mu}^{[-1]}(w_{n_u})))$, where $med$ is a function that returns the median of a vector.
We study the SAR guarantees of the PQM depending on the number of quantiles that the mechanism designer can query.
We divide our investigation into four information settings determined by the value of $k$.
For each information setting, we provide an upper and lower bound on the SAR attainable by truthful mechanisms and characterize the PQM attaining the minimal SAR.
\begin{enumerate}
    \item First, we study the \textit{zero information case}, in which the mechanism designer cannot query any quantile of $\mu$.
    We show that placing the facility at the median of the agents' reports defines a truthful routine whose SAR depends only on the ratio between the number of agents reporting their position and the total number of agents $\lambda=\frac{n_r}{n}$. Furthermore, no other truthful mechanism can achieve a lower SAR. Hence, this median mechanism attains the lowest SAR possible.
    \item Second, we consider the \textit{$n_u$-quantiles information case}, where the mechanism designer can query for at least $n_u$ quantiles of $\mu$ and use such information to elicit the facility position.
    We show that in this case, the PQM induced by $\vec q=(\frac{1}{2n_u}\frac{3}{2n_u},\dots,\frac{2n_u-1}{2n_u})\in[0,1]^{n_u}$ defines a truthful mechanism that attains the optimal cost. That is, its SAR is equal to $1$.
    \item Third, we consider the \textit{median case}, in which the mechanism designer has access only to the median of $\mu$, which we denote with $m$.
    In this case, we consider the PQM induced by $\vec m=(0.5,0.5,\dots,0.5)$ and show that its \textbf{SAR is always lower than $3$}, regardless of $n_r$ and $n_u$.
    We then provide a lower bound on the SAR achievable by any truthful mechanism that elicits the facility position based only on the agents' reports and the median of the probability distribution $\mu$.
    \item Lastly, we consider the \textit{$k$-quantiles case}, in which the mechanism designer has access to $k$ quantiles where $k=2,\dots,n_u$. 
    Since the mechanism designer can query only $k$ quantiles and the PQM needs a $n_u$-dimensional vector, we define the lift operator $L:[0,1]^{k}\to[0,1]^{n_u}$ as
    \begin{equation}
        L:\vec q \to L(q):=(\underbrace{q_1,\dots,q_1}_\text{$t_1$-times},\underbrace{q_2,\dots,q_2}_\text{$t_2$-times},\dots,\underbrace{q_{k-1},\dots,q_{k-1}}_\text{$t_{k-1}$-times},\underbrace{q_k,\dots,q_k}_\text{$t_k$-times}),
    \end{equation}
    where $t_j$ is the number of elements of $\RRn$ whose closest entry of $\vec q$ is $q_j$.
    Notice that $L$ gives a natural way to identify a PQM to any $k$-dimensional vector $\vec q\in[0,1]^k$, that is $\vec q\to\PQM_{L(\vec q)}$.
    We then study the SAR of $\PQM_{L(\vec q)}$ and show that $\SAR(\PQM_{L(\vec q)})=1+\frac{4(1-\lambda)\Dnn(L(\vec q))}{1-2(1-\lambda)\Dnn(L(\vec q))}$, where $\Dnn(\vec w)=\max_{j\in\RRn}\min_{l\in[n_u]}\big|w_l-\frac{2j-1}{2n_u}\big|$.
    Owing to this characterization, we infer that \begin{enumerate*}[label=(\roman*)]
        \item for every $\vec q\in[0,1]^{n_u}$, $\PQM_{L(\vec q)}$ is the PQM with the lowest SAR induced by a $n_u$-dimensional vector whose the entries are belong to $\{q_j\}_{j\in [k]}$.
        More formally, we have that $\SAR(\PQM_{L(\vec q)})\le \SAR(\PQM_{\vec w})$, for every $\vec w\in[0,1]^{n_u}$ such that, given $i\in [n_u]$ there exists a $j\in [k]$ such that $w_i=q_j$.
        \item The best $k$ quantiles that the mechanism designer can query for are the ones minimizing $\Dnn(L(\vec q))$. Since every quantile identified by the set $\RRn$ is such that $q_{i+1}-q_i=\frac{1}{n_u}$, we characterize the vector $\vec q$  explicitly.
        Lastly, we propose a lower bound of the SAR of any mechanism that has access to the $k$ quantiles of $\mu$ induced by a vector $\vec q\in[0,1]^k$ whose entries are equi-distanced, that is $q_j=\frac{2j-1}{2k}$.
    \end{enumerate*}

\end{enumerate}

In Table \ref{tab:introres}, we summarize our results in terms of upper and lower bounds.
In conclusion, we extend our framework and findings to situations requiring the location of two facilities capable of accommodating $c$ agents.
In each scenario, we either define a truthful mechanism with a bounded SAR or demonstrate that no truthful mechanism can achieve a bounded SAR.
Due to space limits, part of the proofs and our results on the FLPAA with two facilities are deferred to the appendix.

\begin{table}[t!]
    \centering
    \begin{tabular}{c@{\hskip 0.7in} c@{\hskip 0.7in} c}
    \specialrule{2.5pt}{1pt}{1pt}
         & Lower Bound & Upper Bound \\
         % \hlin
         \specialrule{1.5pt}{1pt}{1pt}
          $k=0$ & $\frac{2}{\lambda}-1$ & $\frac{2}{\lambda}-1$ \\ 
         \hline
        % \specialrule{0.1pt}{1pt}{1pt}

         $k=1$  & $\begin{cases}
            \max\big\{\frac{4}{1+\lambda},2\big\}-1\quad\quad&\text{if}\quad \lambda\ge\frac{1}{3} \\
            1+\frac{2\lambda}{1-\lambda} \quad\quad &\text{otherwise} 
        \end{cases}$ & $\begin{cases}
           \max\big\{\frac{2}{\lambda + \frac{1}{n}},2\big\}-1\quad\quad &\text{if} \;\; \lambda\ge \frac{1}{2} \\
            1+\frac{2\lambda}{1-\lambda}\quad\quad\quad\quad &\text{otherwise} 
            \end{cases}$ \\ 
         \hline
         
         $1< k< n_u$  & $1+\dfrac{2(1-\lambda)\sigma}{(1+\lambda)n_u+(1-\lambda)\sigma}$ & $1+\dfrac{2(1-\lambda)(\sigma-1)}{n_u-(1-\lambda)(\sigma-1)}$ \\ 
         
         \hline

         $k\ge n_u$ & 1  & 1 \\ 
         \specialrule{2.5pt}{1pt}{1pt}
    \end{tabular}
    \caption{A table containing the upper and lower bounds on the Strong Approximation Ratio for all the four frameworks we consider. 
    The value $n$ represents the number of agents the facility can serve, $n_r$ the number of agents reporting their position, and $n_u:=n-n_r$. 
    The value $\lambda=\frac{n_r}{n}$ is the fraction of agents reporting their position. 
    For the sake of simplicity,  in the case $1<k<n_u$, we assume that $k$ divides $n_u$, that is $n_u=\sigma\, k$ where $\sigma\in \enne$.
    Moreover, the lower bound we reported is restricted to the case in which $k$ is even and the quantiles are induced by a vector $\vec q\in[0,1]^k$ whose entries are equi-distanced, that is $q_j=\frac{2j-1}{2k}$.}
    \label{tab:introres}
\end{table}

\paragraph*{Related Works.}

The Facility Location Problem (FLP) and its variations are significant issues in various practical domains, such as disaster relief \cite{doi:10.1080/13675560701561789}, supply chain management \cite{MELO2009401}, healthcare \cite{ahmadi2017survey}, clustering \cite{hastie2009elements}, and public facilities accessibility \cite{barda1990multicriteria}.
Procaccia and Tennenholtz initially delved into the Mechanism Design study of the $m$-FLP, laying the groundwork for this field in their pioneering work \cite{procaccia2013approximate}.
Following that, a range of mechanisms with constant approximation ratios for placing one or two facilities on trees \cite{DBLP:conf/sigecom/FeldmanW13,DBLP:conf/atal/FilimonovM21}, circles \cite{DBLP:conf/sigecom/LuSWZ10,DBLP:conf/wine/LuWZ09}, general graphs \cite{10.2307/40800845,DBLP:conf/sigecom/DokowFMN12}, and metric spaces \cite{DBLP:conf/sagt/Meir19,DBLP:conf/sigecom/TangYZ20} were introduced.
Despite the generality of the underlying space, it is important to stress that all these positive results are confined to the case in which we have at most two facilities to place and/or the number of agents is limited.
Moreover, different works tried to generalize the initial framework proposed in \cite{procaccia2013approximate}, by considering different agents' preferences \cite{church1978locating,MEI201846}, different costs \cite{feigenbaum2017approximately,cai2016facility}, and additional constraints \cite{zou2015facility,feldman2016voting}.
% 
% We refer the readers to \cite{chan2021mechanism} for a comprehensive survey of the mechanism design aspects of the FLP.
% 

% 
The $m$-Capacitated Facility Location Problem ($m$-CFLP) is a variant of the $m$-FLP in which each facility can accommodate a finite number of agents.
The Mechanism Design aspects of the $m$-CFLP have only recently begun to attract attention.
Indeed, the first game theoretical framework for the $m$-CFLP was introduced in \cite{aziz2020facility}.
This work defined and studied various truthful mechanisms, like the InnerPoint Mechanism, the Extended Endpoint Mechanism, and the Ranking Mechanisms.
A more theoretical study of the problem was later presented in \cite{ijcai2022p75}, demonstrating that no mechanism can place more than two capacitated facilities while adhering to truthfulness, anonymity, and Pareto optimality.
Notice that, by dropping Pareto optimality, it is possible to define truthful mechanisms capable of placing more than two facilities and has bounded approximation ratios \cite{auricchio2024facility}.
Lastly, papers that deal with different Mechanism Design aspects of the $m$-CFLP are \cite{auricchio2023extended}, where the $m$-CFLP is studied in a Bayesian setting, and \cite{aziz2020capacity,auricchiocapacitated}, where the authors investigate the case in which the facility to place cannot accommodate all the agents.
To some extent, our framework is similar to one used in Bayesian Mechanism design, in which the mechanism designer has to design routines to serve a population distributed according to $\mu$.
Bayesian Mechanism Design has been applied to investigate routing games \cite{gairing2005selfish}, facility location problems \cite{zampetakis2023bayesian,auricchio2023extended}, combinatorial mechanisms using $\epsilon$-greedy mechanisms \cite{lucier2010price}, and, notably, auction mechanism design \cite{chawla2009sequential,chawla2007algorithmic,hartline2009simple,yan2011mechanism,ensthaler2014bayesian}.
However, our framework distinguishes itself from Bayesian mechanism design for two reasons.
\begin{enumerate*}[label=(\roman*)]
    \item Unlike what happens in Bayesian mechanism design, in our case the mechanism designer does not know the probability distribution $\mu$, but only some qualitative information of $\mu$, \textit{e.g.} the quantiles.
    \item In Bayesian mechanism design, the performances of the mechanisms are measured in expectation, while in our case, the SAR is defined as a worst-case ratio.
\end{enumerate*}
Finally, a few works study and propose a model for the FLP in which it involves a degree of uncertainty.
In this case, the agents' preferences and locations are based on distributional assumptions.
In \cite{caragiannis2016truthful}, the authors explored single-facility locations in settings where agents' locations are independently and identically drawn from an unknown distribution. 
Another different approach is presented in \cite{menon2019mechanism}.
In this case, the authors addressed a scenario where each agent is associated with an interval on the line representing all possible locations, delving into robust mechanisms designed to perform well across all potential unknown true preferred locations within those intervals. 

\section{Setting Statement}
In what follows, we assume that agents and the facility are placed on a line.
Let $n$ denote the total amount of agents that the facility can serve, $n_r$ denote the number of agents reporting their position, and $n_u=n-n_r$ denote the spare capacity of the facility used by the agents belonging to the population $\mu$.
We name \textit{deterministic agents} the agents reporting their position, while the other agents are called \textit{aleatory agents}. 
Lastly, let us denote with $\PP(\erre)$ the set of the probability measures supported over $\erre$.
Given the position of the $n_r$ deterministic agents, our goal is to place a facility in a position that minimizes the combined costs of the $n_r$ deterministic agents and the expected costs of $n_u$ i.i.d. aleatory agents distributed according to $\mu\in\PP(\erre)$.
Given a position $y\in\erre$, a deterministic agent located at $x_i$ incurs a cost of $c_i(x_i,y)=|x_i-y|$ to access the facility, while an aleatory agent sampled from $X\sim\mu$ incurs in an ex-ante cost equal to $c(y,\mu)=\EE_{X\sim\mu}[|X-y|]$.

\begin{problem}
\label{problem:fac1}
Let $n=n_r+n_u$ be the capacity of a facility.
Given $\vec x\in\erre^{n_r}$, the Facility Location Problem with Aleatory Agents (FLPAA) associated with $\vec x$ and $\mu$ consists in finding the $y^*\in\erre$ that minimizes the ex-ante social cost function, namely $\ESC$, that is
\begin{equation}
    \label{eq:defproblemonefacility}
    y\to \mathcal{ESC}(\vec x, y ; \mu)= \sum_{i=1}^{n_r}c_{i}(x_i,y)+\sum_{j=1}^{n_u}c_j(y,\mu)=\sum_{i=1}^{n_r}|x_i-y|+n_u\EE_{X\sim\mu}[|X-y|].
\end{equation}
\end{problem}

\paragraph*{Basic Assumptions.}

Finally, we layout the two basic assumptions of our framework.
In what follows, we tacitly assume that the underlying distribution $\mu$ satisfies the two following properties:
\begin{enumerate*}[label=(\roman*)]
\item $\mu\in\PP(\erre)$ has finite first moment, i.e. $\int_\erre |x|d\mu < +\infty$. This condition is essential, as otherwise the expected ex-ante cost of the aleatory agents would be not finite.
\item The measure $\mu$ is absolutely continuous. We denote with $\rho_\mu$ its probability density function. This assumption is not essential, but it allows us to simplify the discussion. 
Indeed, up to an arbitrary small error, every probability measure can be approximated by an absolutely continuous probability measure, \cite{villani2009optimal}.
Throughout the paper, we make extensive use of this property to approximate discrete probability measures.
For the sake of simplicity, we say that a sequence of probability measures $\mu_\ell$ \textit{concentrates the probability at} one or more points $a\in\mathcal{A}\subset \erre$ as $\ell\to\infty$ if $\mu_\ell$ converges to a discrete probability measure supported over $\mathcal{A}$.
A standard example, is given by the sequence $\mu=2\ell\mathcal{U}_{[x-\frac{1}{\ell},x+\frac{1}{\ell}]}$, which converges to a probability measure that assigns probability $1$ to $x\in \erre$, that is the sequence $\mu_\ell$ concentrates all the probability at $x$.
In Figure \ref{fig:battleofvectors_firsttest}, we give a graphical description of what concentrating the probability means in practice.
\end{enumerate*}
Lastly, notice that, according to this set of assumptions, the cumulative distribution function (c.d.f.) of $\mu$, namely $F_\mu$, is locally bijective.
%
% and thus no quantiles have the position on the line.

\begin{figure}[t!]
  \centering
  \begin{subfigure}{0.4\linewidth}
  \centering
    \includegraphics[width=\linewidth]{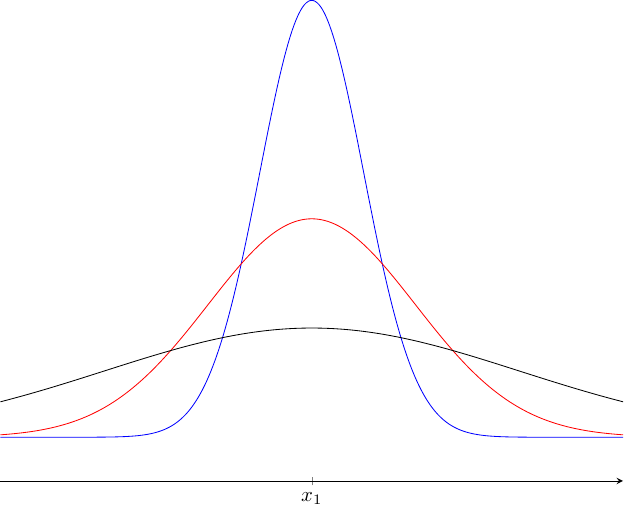}
    
  \end{subfigure}
  \hspace{15mm}
  \begin{subfigure}{0.4\linewidth}
    \centering
    \includegraphics[width=\linewidth]{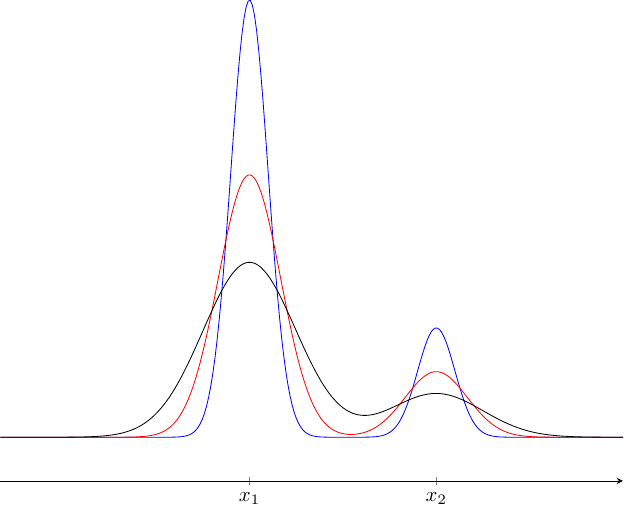}
    % \caption{Subfigure 2}
  \end{subfigure}

  \caption{Two sequences of probability measures that concentrate the probability at one and two points, respectively.
The sequence of measures on the left concentrates all the probability at $x_1$, and in the limit, it converges to a probability measure whose random variable is equal to $x_1$ with a probability of $1$.
On the right, the sequence concentrates the probability at $x_1$ and $x_2$. In the limit, it converges to a probability measure whose associated random variable is equal to $x_1$ with a probability of $0.8$ or equal to $x_2$ otherwise.
\label{fig:battleofvectors_firsttest}}
\end{figure}

\subsection{Computing the Optimal Solution of the FLPAA}

In this section, we study the optimal solution to Problem \ref{problem:fac1}.
Given a vector $\vec x=(x_1,\dots,x_{n_r})$, the \textit{empirical cumulative distribution function} (ecdf) associated with $\vec x$ is defined as $F_{\vec x}(t)=\frac{1}{n_{r}}\sum_{i=1}^{n_r}\II_{\{x_i\le t\}}(t)$, where $\II_{\{x_i\le t\}}(t)$ is the indicator function of the set $\{x_i\le t\}$, which is equal to $1$ if $x_i\le t$ and $0$ otherwise.
Given $\frac{n_r}{n}=\lambda\in[0,1]$, we set
\begin{equation}
\label{eq:flam}
    F_{\lambda,\mu,\vec x}(t)=\lambda F_{\vec x}(t) + (1-\lambda)F_{\mu}(t).
\end{equation}
Since $F_{\lambda,\mu,\vec x}$ is the convex combination of two c.d.f., we have that $F_{\lambda,\mu,\vec x}$ is a c.d.f. as well.
As a first result, we characterize the set containing the optimal solutions to the FLPAA.

\begin{theorem}
\label{thm:problem_opt}
    Given $\vec x\in\erre^{n_r}$ and $\mu\in\PP(\erre)$, every median of $\Flam$ is a solution, thus $y^*=\inf\big\{y\in\erre,\;\;\text{s.t.}\;\; F_{\lambda,\mu,\vec x}(t)\ge \frac{1}{2}\big\}$ is a solution.
    Moreover, given $\vec x\in\erre^{n_r}$ and $\mu\in\PP(\erre)$, $y\in\erre$ is a solution to Problem \ref{problem:fac1} if and only if $y\in[a,b]$ where $a=\sup\{t\in\erre,\,\text{s.t.}\,\,\partial_y\mathcal{ESC}(t)< 0\}$ and $b=\inf\{t\in\erre,\,\text{s.t.}\,\,\partial_y\mathcal{ESC}(t) > 0\}$.
    Finally, if none of the elements in $\mathcal{X}:=\{x_i\}_{i\in[n_r]}$ is an optimal solution, then the optimal solution $y^*$ is unique, $y^*$ belongs to the support of $\mu$, and $\partial_y\ESC(y^*)=0$.
\end{theorem}

From Theorem \ref{thm:problem_opt} we infer that the FLPAA has at least a solution.
Moreover, we show that there exists a discrete set that contains at least one solution to the FLPAA and characterize it.

\begin{theorem}
\label{thm:setoffeasibleopt}
    Given $n,n_r,n_u$ and $\mu$, let us denote with $\mathcal{F}_{n_u}=\{f_j\}_{j\in[n_u]}$ where $f_j=F^{[-1]}_\mu(\frac{2j-1}{2n_u})$ if $n$ is odd and $f_j=F^{[-1]}_\mu(\frac{j}{n_u})$ if $n$ is even.
    Then, for any given $\vec x\in\erre^{n_r}$, at least one element of the set $\mathcal{X}\cup\mathcal{F}_{n_u}$ is an optimal solution to Problem \ref{problem:fac1}, where $\mathcal{X}=\{x_i\}_{i\in[n_r]}$ is the set of containing all the agents report.
    In particular, Problem \ref{problem:fac1} is solvable in polynomial time.
\end{theorem}

Notice that the $q$ values defining $\mathcal{F}_{n_u}$ do not depend on $\mu$.
Indeed, every $f_j\in\mathcal{F}_{n_u}$ is the quantile associated to the value $\frac{2j-1}{2n_u}$ regardless of the probability measure $\mu$. 
By refining the proof of Theorem \ref{thm:setoffeasibleopt}, it is possible to restrict the set containing the possible optimal solutions whenever $n_r<n_u$.

\begin{corollary}
\label{cor:essentialpoints}
    If $n_r<n_u$, we have that the set $\mathcal{X}\cup\tilde{\mathcal{F}}_{n_u}$, where $\tilde{\mathcal{F}}_{n_u}=\{f_j\}_{j\in\{\alpha,\dots,\beta\}}$, $\alpha=\frac{n+1}{2}-n_r$, and $\beta=n_u-\alpha$ contains at least an optimal solution.
\end{corollary}

Lastly, we define the set of relevant quantiles $\RRn$ as the set of indexes $j\in[n_u]$ for which there exists $\vec x\in\erre^{n_r}$ and $\mu\in\PP(\erre)$ such that the optimal solution to Problem \ref{problem:fac1} is $F_\mu^{[-1]}(\frac{2j-1}{2n_u})$.
More formally, when $n$ is odd, we have
\begin{equation}
\label{def:rrn}
    \RRn:=\bigg\{j\in[n_u]\;\;\text{s.t.}\;\;\exists \vec x\in\erre^{n_r},\;\; F_\mu^{[-1]}\Big(\frac{2j-1}{2n_u}\Big)\in\argmin_{t\in\erre}\Big\{ F_{\lambda,\mu,\vec x}(t)\ge \frac{1}{2}\Big\}\bigg\}.
\end{equation}
Similarly, we define $\RRn$ for $n$ even, in both cases, that the cardinality of $\RRn$ is finite.

\subsection{The Mechanism Design Aspects for the FLPAA} 
A mechanism is a function $M:\erre^{n_r}\to\erre$ that takes in input the reports of the $n_r$ deterministic agents $\vec x\in\erre^{n_r}$ and outputs a facility position $y=M(\vec x)\in\erre$.
A mechanism $M$ is said to be \textit{truthful} (or \textit{strategy-proof}) if the cost of a deterministic agent is minimized when it reports its true position. 
That is, $c_i(x_i,M(\vec x))\le c_i(x_i,M(\vec x_{-i},x_i'))$ for any misreport $x_i'\in \erre$, where $\vec x_{-i}$ is the vector $\vec x$ without its $i$-th component.
Although deploying a truthful mechanism prevents agents from getting a benefit by misreporting their positions, this leads to a loss of efficiency.
{\color{black} Again, we stress that, while the mechanism has to be truthful with respect to the set of $n_r$ deterministic agents as they are the one reporting the information before the facility gets located.
The remaining $n_u$ agents do not submit their preferences to the mechanism as thus cannot manipulate the outcome of the eliciting process.}
To evaluate this efficiency loss, Nisan and Ronen introduced the notion of approximation ratio \cite{nisan1999algorithmic}.
Given a truthful mechanism $M$ and a probability distribution $\mu$, the approximation ratio of $M$ with respect to the ex-ante social cost is 
\begin{equation}
    % \label{eq:worstcaseapproximationratio}
    ar_\mu(M):=\sup_{\vec x \in \mathbb{R}^{n}}\frac{\ESC_{M}(\vec x;\mu)}{\ESC_{opt}(\vec x;\mu)},
\end{equation}
where $\ESC_M(\vec x;\mu)$ is the ex-ante social cost of placing the facility at $M(\vec x)$, while $\ESC_{opt}(\vec x;\mu)$ is the optimal ex-ante social cost achievable when the agents' report is $\vec x$.
Notice that the approximation ratio of $M$ does depend on the probability measure $\mu$.
We are interested in defining routines that do not depend on specific $\mu$, but that work well across every possible probability distribution.
For this reason, we consider the notion of strong approximation ratio (SAR), which measures the ratio of the cost attained by a mechanism and the optimal cost over every possible input $\vec x$ and over every possible distribution $\mu$, that is
\begin{equation}
    \label{eq:SAR}
    \SAR(M):= \sup_{\mu\in\PP(\erre)}ar_\mu(M)=\sup_{\mu\in\PP(\erre)}\sup_{\vec x \in \mathbb{R}^{n}}\frac{\ESC_{M}(\vec x;\mu)}{\ESC_{opt}(\vec x;\mu)}.
\end{equation}
Notice that the SAR is a stricter metric than the classic approximation ratio since, for every truthful mechanism $M$, $ar_\mu(M)\le \SAR(M)$.
In our study, we consider the scenario in which the mechanism designer does not know the probability distribution $\mu$, but it can access a $k$-tuple of quantiles of $\mu$, which are decided before the mechanism designer receives the agents' reports.
The challenge in this framework is to define the best truthful mechanism $M$ in terms of SAR based on the information on $\mu$ that the mechanism designer can gather.
To this extent, we consider the class of \textit{Phantom Quantile Mechanism} (PQM).
\begin{definition}
    Given $n_u\in\enne$, let $\vec q\in[0,1]^{n_u}$ be such that $q_i\le q_{i+1}$.
    Then, the Phantom Quantile Mechanism (PQM) associated with $\vec q$ is defined as $\PQM_{\vec q}(\vec x)=med(\vec x, \vec f)$, where $f_k=F_{\mu}^{[-1]}(q_k)$ for every $k\in[n_u]$.
\end{definition}
It is easy to see that the PQM is always truthful and anonymous as, for any given $\mu$, they are instances of the Phantom Peak Mechanisms, introduced in \cite{moulin1980strategy}.

\begin{theorem}
    \label{them:truthfulPQM}
    $\PQM_{\vec q}$ is truthful and anonymous for every $\vec q \in [0,1]^{n_u}$ and every $\mu\in\PP(\erre)$.
\end{theorem}

Notice that the PQM defines a class of mechanisms, not a single mechanism, which brings us to the two main issues we address in the following: 
\begin{enumerate*}[label=(\roman*)]
    \item to determine the minimum number of quantiles, namely $k_{min}$, needed to define an optimal mechanism and
    \item to determine the best quantiles to query for when the number of quantiles the mechanism designer can query is less than $k_{min}$. 
\end{enumerate*}  

\section{Tailor-making a Mechanism to different Levels of Information}

We study how to define truthful mechanisms that, given in input a vector containing the agents' reports, elicit the facility position based on the input $\vec x$ and the available information on the distribution $\mu$.
We assume that the mechanism designer can query $k$ quantiles from $\mu$, which they can use to improve the performances of the mechanism.
Therefore, in this framework, the parameter $k$ quantifies the mechanism designer's degree of knowledge of the probability distribution $\mu$.
We assume that the $k$ quantiles are queried before the mechanism designer receives the agents' reports.
This is a key assumption, as it means that the $q$ values associated with the quantiles do not depend on the vector $\vec x$.
We study this problem for every $k$, however, for the sake of the exposition, we divide our study as follows:
\begin{enumerate*}[label=(\roman*)]
    \item in the \textit{zero information case}, in which $k$ is equal to $0$;
    \item the \textit{$n_u$-quantile information case}, in which $k\ge n_u$;
    \item the \textit{median information case}, in which the mechanism designer has access to the median of $\mu$, \textit{i.e.} $k=1$; and
    \item the \textit{$k$-quantile information case}, in which the mechanism designer has access to any $1<k<n_u$ quantiles of $\mu$.
\end{enumerate*}
For the sake of simplicity, from now on, we assume that the number of agents $n$ that the facility can serve is odd.
All our results are extendable to the case in which $n$ is even in a trivial way.

\subsection{Zero Information case}

In this case, the mechanism designer does not have access to any information on the probability distribution $\mu$, therefore the mechanism elicits the position of the facilities based solely upon the vector reporting the agents' positions, that is $\vec x\in\erre^{n_r}$.
In this case, the median mechanism that places the facility at the median of $\vec x$ defines a truthful mechanism that achieves bounded SAR.

\begin{theorem}
\label{thm:SAR0inf}
    The mechanism $M:\vec x \to med(\vec x)$ is truthful.
    Moreover, we have that
    \begin{equation}
        \SAR(M)=\begin{cases}
            \frac{2n_u+n_r-1}{n_r+1}=\frac{2-\lambda-\frac{1}{n}}{\lambda+\frac{1}{n}}\quad\quad &\text{if}\quad n_r \; \text{is odd},\\
            % \\
            \frac{2n_u+n_r}{n_r}=\frac{2-\lambda}{\lambda}\quad\quad &\text{if}\quad n_r \; \text{is even.}
        \end{cases}
    \end{equation}
\end{theorem}

\begin{proof}[Sketch of the Proof]
    For the sake of argument, let us consider $n_r$ even.
    Consider an instance in which half of the reporting agents are located at $0$ and the other half is located at $1$.
    By definition, the mechanism places the facility at $0$, therefore, the worst instance happens when the distribution $\mu$ is supported on an interval that is small and close to $1$.
    By doing so, we ensure that the $n_u$ aleatory agents are all close to $1$, while the facility is wrongly placed at $0$.
    By taking a sequence of probability measures that concentrate the mass at $1$ in the limit, we have that all the aleatory agents are located at $1$.
    A simple computation shows that the cost of the mechanism is then $n_u+\frac{n_r}{2}$, while the optimal cost is $\frac{n_r}{2}$.
    This proves that the SAR of the mechanism is larger than $\frac{n+n_u}{n_r}=\frac{2-\lambda}{\lambda}$.
\end{proof}

We complement Theorem \ref{thm:SAR0inf} with an example showcasing how the approximation ratio of a mechanism $M$ for a fixed measure $\mu$ is smaller than the SAR of the mechanism.

\begin{example}
    Let us fix $n=5$, $n_r=3$, and $n_u=2$.
    Owing to Theorem \ref{thm:SAR0inf}, the mechanism $M:\vec x\to med(\vec x)$ has a SAR equal to 
    $\frac{3}{2}$.
    Let us now consider the case in which $\mu$ is the uniform distribution on the interval $[1,2]$, that is $\mu=\mathcal{U}([1,2])$.
    Following the proof of Theorem \ref{thm:SAR0inf}, we infer that the worst case instance is $\vec x=(0,0,1.25)$, since $F_\mu^{[-1]}(\frac{1}{4})=1.25$.
    By a simple computation, we have $\ESC_{opt}(\vec x)=\frac{15}{4}$, while $\ESC_{M}(\vec x)=\frac{17}{4}$, hence $ar_{\mathcal{U}([1,2])}(M)=\frac{17}{15}\sim 1.13<\frac{3}{2}$.
\end{example}

% % 
We close this section by showing that placing the facility at the median of the agents' reports is the best possible mechanism for this framework.

\begin{theorem}
\label{thm:app}
    No truthful mechanism can achieve a SAR that is lower than $\frac{2-\lambda-\frac{1}{n}}{\lambda+\frac{1}{n}}$ if $n_r$ is odd or lower than $=\frac{2-\lambda}{\lambda}$ if $n_r$ is even. 
\end{theorem}

\begin{proof}[Sketch of the Proof]
    Consider an instance in which half of the reporting agents are located at $0$ and the other half is located at $1$.
    Since $M$ is truthful, we can assume without loss of generality that the mechanism places the facility at a location $y$ that is either $0$ or $1$.
    Indeed, otherwise, we can move all the agents at $0$ to $y$ and, up to a scaling factor, the argument follows the same.
    If the mechanism places the facility at $0$, the proof follows by considering a sequence of probability measures that concentrates the mass at $1$.
    If the mechanism places the facility at $1$, we proceed similarly.
    By computing the two possible ratios and taking the minimum, we conclude the proof.
\end{proof}

\subsection{\texorpdfstring{$n_u$}{k}-Quantile Information Case}

When $k\ge n_u$, it is possible to define an optimal and truthful mechanism.
Moreover, the routine of the optimal mechanism does not depend on the value of $n_r$.
Indeed, given $n_u$, the PQM induced by the vector $\vec q=(\frac{1}{2n_u},\frac{3}{2n_u},\dots,\frac{2n_u-1}{2n_u})$, is both truthful and optimal, hence placing the facility at the median of the vector $(\vec x,\vec f)$, defines an optimal routine.

\begin{theorem}
\label{thm:opt_discretemech}
    The mechanism $\PQM_{\vec q}$ where $q_k=\frac{2k-1}{2n_u}$ for every $k\in[n_u]$ is truthful and optimal regardless of $\mu$, that is $\SAR(M)=1$.
\end{theorem}

\begin{proof}[Sketch of the Proof]
    The proof of this theorem heavily relies upon the characterization presented in Theorem \ref{thm:problem_opt}.
    Given $\vec x\in\erre^{n_r}$, the ex-ante social cost $\ESC$ is a differentiable function with respect to $y$.
    Moreover, we have $\partial_{y}\ESC(\vec x, y;\mu)=n\big(\lambda \frac{2\# \{x_i\le y\}-n_r}{n_r}+2(1-\lambda)F_\mu(y)\big)$.
    Since $\PQM_{\vec q}$ with $\vec q=(\frac{1}{2n_u},\frac{3}{2n_u},\dots,\frac{2n_u-1}{2n_u})$ places the facility at the median of $(\vec x,(F_\mu^{[-1]}\big(\frac{1}{2n_u}\big),\dots,F_\mu^{[-1]}\big(\frac{2n_u-1}{2n_u}\big)))$, it is possible to evaluate the value of $\partial_y\ESC$ on the output of $\PQM_{\vec q}$. 
    Indeed, by definition of $\PQM_{\vec q}$, we have that if $\PQM_{\vec q}(\vec x)=\Flam^{[-1]}(\frac{2l-1}{2n_u})$ then $\{x_i\le y\}=\frac{n+1}{2}-l$, hence $\partial_y\ESC(\PQM_{\vec q}(\vec x))$.
    Through a similar argument, we handle the case in which $\PQM_{\vec q}(\vec x)=x_i$ where $i\in[n_r]$.
\end{proof}

Notice that, as a consequence, we infer that if the mechanism designer has access to the complete distribution $\mu$, they can design a truthful and optimal mechanism.

\subsection{Median Information Case}
\label{sec:justmedian}
We now consider the case in which the mechanism designer has only access to the median of the measure $\mu$, which we denote with $m$.
In this case, we consider the PQM induced by the vector $\vec q=(0.5,0.5,\dots,0.5)\in[0,1]^{n_u}$, that is $\PQM_{\vec q}(\vec x)=med(\vec x,\vec m)$, where $\vec m=(m,m,\dots,m)\in\erre^{n_u}$.
This mechanism is truthful and has bounded SAR, however, its SAR is always larger than $1$ when $1<n_u<n-1$.

\begin{theorem}
\label{thm:armedianmech}
    Let $\vec q=(0.5,0.5,\dots,0.5)$, where $m$ is the median of $\mu$.
    The mechanism $\PQM_{\vec m}$ is optimal if and only if $n_u\in\{0,1,n\}$.
    In all other cases, we have
    \[
        \SAR(\PQM_{\vec q})=\begin{cases}
            1+\max\big\{\frac{2n}{n-n_u-1},0\big\}=\max\big\{\frac{2}{\lambda + \frac{1}{n}},2\big\}-1\quad\quad &\text{if} \;\; \lambda\ge \frac{1}{2} \\
            1+\frac{2n}{n_u}=1+\frac{2\lambda}{1-\lambda}\quad\quad\quad\quad &\text{otherwise.}
        \end{cases}
    \]
    In particular, $\SAR(\PQM_{\vec q}) \le 3$.
\end{theorem}

\begin{proof}[Sketch of the Proof]
    For the sake of argument, let us consider the case in which $\lambda>\frac{1}{2}$ or, equivalently, $n_r>n_u$.
    Given $\mu$ and $\vec x$, if the optimal facility location is $m$, then the mechanism returns the optimal location.
    We can then assume that the optimal facility location is $y^*<y\le m$, where $y$ is the output of the mechanism.
    By carefully handling the ratio between the cost of the mechanism and the optimal cost, we show that the worst case happens when
    \begin{enumerate*}[label=(\roman*)]
        \item the output of the mechanism is $m$,
        \item all the agents are located at either $y^*$ or $m$, and
        \item there are $\frac{n-1}{2}$ agents at $y^*$ and the others are located at $m$.
    \end{enumerate*}
    Owing to Theorem \ref{thm:problem_opt}, we infer that $y^*=F_\mu^{[-1]}(\frac{1}{2n_u})$.
    We then consider a sequence of probability measures that concentrate half of the probability at $y^*$ and the remaining at $m$.
    In this way, it is possible to show that the $\ESC$ of the mechanism converges to $\frac{n_u}{2}+\frac{n-1}{2}$, while the optimal cost converges to $n_r-\frac{n-1}{2}+\frac{n_u}{2}$, which concludes the proof.
\end{proof}

To conclude this section, we prove a lower bound on the SAR for any truthful mechanism that has access to the median of $\mu$.

\begin{theorem}
\label{thm:LBmedian}
    Let $M$ be a truthful mechanism that has access to the median of $\mu$, then
    \[
        \SAR(M)\ge\begin{cases}
            \max\big\{\min\big\{ \frac{n}{\floor{\frac{n+n_r}{4}}+1},\frac{2n}{2n-2\floor{\frac{n+n_r}{4}}-n_u} \big\},2\big\}-1\quad\quad &\text{if}\quad \lambda\ge\frac{1}{3} \\
            \frac{2\lambda}{1-\lambda}+1 \quad\quad &\text{otherwise.}
        \end{cases}
    \]
    If we take the limit for $n\to \infty$, we obtain
    \[
        \SAR(M)\ge\begin{cases}
            \max\big\{\frac{4}{1+\lambda},2\big\}-1\quad\quad &\text{if}\quad \lambda\ge\frac{1}{3} \\
            \frac{2\lambda}{1-\lambda}+1 \quad\quad &\text{otherwise.}
        \end{cases}
    \]
    In particular, if $\lambda\ge \frac{n-1}{n}$, the lower bound is $1$ and it is matched by $\PQM_{\vec q}$ where $\vec q=(0.5,\dots,0.5)$.
    % \footnote{Please notice that, for the sake of simplicity, we gave asymptotic bound attained when the number of agents goes to infinity. The right lower bound can be recovered by substituting $\lambda$ with $\frac{\floor{\lambda n}}{n}$.}
\end{theorem}

\begin{proof}[Sketch of the Proof]
    For the sake of simplicity, we consider the case in which $n\to \infty$.
    The proof of this theorem is divided into two parts, depending on whether $\lambda$ is greater or smaller than $\frac{1}{3}$.
    We sketch how to handle the case in which $\lambda$ is larger than $\frac{1}{3}$ since the other case is easier.
    Let $m$ be the median of $\mu$ and let $M$ be a truthful mechanism that has access to the median of $\mu$.
    We consider an instance where the agents are either located at $t<m$ or at $m$.
    The amount of agents located at $t$ (and consequentially the ones located at $m$) are determined as follows.
    Let $l$ be the number of agents located at $t$ and let $y$ be the position returned by the mechanism on this input.
    If $y>m$, we use the truthfulness of $M$ to move all the agents to $m$ without altering the outcome of the mechanism. 
    We can then take a sequence of probability measures whose median is $m$ and whose support is a small interval containing $m$ (consider a uniform distribution as $2\epsilon\mathcal{U}_{[m-\epsilon,m+\epsilon]}$ where $\epsilon$ is a small and positive value) and show that the ratio between the mechanism cost and the optimal cost diverges.
    If $y<m$, we can again use truthfulness to move all the agents from $t$ to $y$, hence we can assume $t=y$ in this case.
    Therefore, we have to consider only two cases: either the mechanism places the facility at $t$ or the mechanism places the facility at $m$.
    If the mechanism places the facility at $t$, we consider a sequence of probability measures that concentrates all the probability on a small interval that contains $m$.
    If the mechanism places the facility at $m$, we take a sequence of probability distribution that places half probability at $t$ and half probability at $m$.
    In both cases, the ratio between the mechanism cost and the optimal cost depends on the number of agents located at $l$; hence, to retrieve our bound, we look for the value of $l$ that maximizes the minimum of the two ratios.
\end{proof}

\subsection{\texorpdfstring{$k$}{k}-Quantile Information Case}
We now consider the case in which the mechanism designer has access to $k< n_u$ quantiles.
In particular, given $\vec q\in[0,1]^k$, we characterize the PQM that achieves the lowest SAR while having access to the quantiles associated with $\vec q$.
We then characterize the optimal $\vec q\in[0,1]^k$ the mechanism designer can query for.
To achieve these results, we need to introduce two mathematical tools.

\begin{definition}
    Given $n_u$, $n_r$, and $k$, we define the lift operator $L:[0,1]^k\to [0,1]^{n_u}$ as follows
    \begin{equation}
        L:\vec q \to L(q):=(\underbrace{q_1,\dots,q_1}_\text{$t_1$-times},\underbrace{q_2,\dots,q_2}_\text{$t_2$-times},\dots,\underbrace{q_{k-1},\dots,q_{k-1}}_\text{$t_{k-1}$-times},\underbrace{q_k,\dots,q_k}_\text{$t_k$-times}),
    \end{equation}
    where $t_j$ is the number of elements of $\RRn$ whose closest entry of $\vec q$ is $q_j$ ($t_j$ can also be equal to $0$).
    In particular, $t_j$ is the number of elements in $j\in \RR(n_r,n_u)$ for which $|q_j-\frac{2j-1}{2n_u}|=\min_{l\in[k]}|q_j-\frac{2l-1}{2n_u}|$ holds.
\end{definition}

\begin{definition}
    Given $n_u$, $n_r$, and $k$, we define $\Dnn:[0,1]^{n_u}\to\erre$ as $\Dnn(\vec w)=\max_{j\in\RRn}\big|w_j-\frac{2j-1}{2n_u}\big|$, where $\RRn$ is the set of relevant indexes defined in \eqref{def:rrn}.
\end{definition}

The lift operator allows us to retrieve a $n_u$ dimensional vector from a $k$ dimensional one, thus, given $\vec q$, the mechanism $\PQM_{L(\vec q)}$ is well defined.
In the following theorem, we compute the SAR of $\PQM_{L(\vec q)}$ and show that it is the best PQM induced by a vector that has the same entries as $\vec q$.

\begin{theorem}
\label{thm:PQMar}
    Given $\vec q$, the mechanism $\PQM_{L(\vec q)}$ is well-defined.
    Moreover, we have that
    \begin{equation}
    \label{eq:ar_fixedq}
        \SAR(\PQM_{L(\vec q)})=1+\frac{4(1-\lambda)\Dnn(L(\vec q))}{1-2(1-\lambda)\Dnn(L(\vec q))}.
    \end{equation}
    Finally, $L(\vec q)$ induces the PQM with the lowest SAR amid the class of PQM induced by vectors whose entries are the same as the entries of $\vec q$, that is $\SAR(\PQM_{L(\vec q)})\le \SAR(\PQM_{\vec w})$, for every $\vec w\in[0,1]^{n_u}$ such that, for every $i\in [n_u]$ there exists a $j\in [k]$ such that $w_i=q_j$.
\end{theorem}

\begin{proof}[Sketch of the Proof]
    Let $j\in\RRn$, then there exists $\vec x$ and $\mu$ such that the optimal solution is $y^*=F_\mu^{[-1]}(\frac{2j-1}{2n_u})$ and $\PQM_{L(\vec q)}(\vec x)=y=F_\mu^{[-1]}(q_j)$.
    Without loss of generality, we assume that $q_j<\frac{2j-1}{2n_u}$, hence $y<y^*$, and set $\Delta_q=\big|q_j-\frac{2j-1}{2n_u}\big|$.
    We then apply a sequence of modifications to the instance in order to maximize the ratio between the mechanism cost and the optimal cost.
    In particular,
    \begin{enumerate*}[label=(\roman*)]
        \item \label{bullet1proofdeltaq} we move all the deterministic agents on the left of $y$ to $y$. Since we are decreasing the optimal and mechanism cost by the same quantity the ratio increases. Similarly, we increase the ratio by considering a sequence of probability measures that concentrate the probability that $\mu$ assigns to the left of $y$ on a small interval close to $y$.  
        \item \label{bullet1proofdeltaq2} We repeat the process in \ref{bullet1proofdeltaq} to the agents and probability on the right of $y^*$. 
        \item \label{bullet1proofdeltaq3} We move all the agents whose position is between $y$ and $y^*$ to $y^*$. Finally, we concentrate all the probability that $\mu$ assigns to $(y,y^*)$ to $y^*$.
    \end{enumerate*}
    Notice that all these modifications do not affect the position of the optimal solution nor the output of $\PQM_{L(\vec q)}$.
    Lastly, we compute the optimal and the mechanism cost.
    With a slight abuse of notation, we denote with $\vec x$ and $\mu$ the agents' reports and the probability measure we obtained after modifying the instance according to points \ref{bullet1proofdeltaq}, \ref{bullet1proofdeltaq2}, and \ref{bullet1proofdeltaq3}.
    By Theorem \ref{thm:problem_opt}, we have that $\Flam(y^*)\ge \frac{1}{2}$.
    By construction, $\Flam(y)=\frac{1}{2}-(1-\lambda)\Delta_q$, therefore the ratio between the mechanism cost and the optimal cost is $\frac{1+2(1-\lambda)\Delta_q}{1-(1-\lambda)\Delta_q}=1+\frac{4(1-\lambda)\Delta_q}{(1-\lambda)\Delta_q}$, which concludes the proof.
\end{proof}

\begin{remark}
    It is worthy of notice that Theorem \ref{thm:PQMar} allows us to define the best PQM for a given vector $\vec q\in[0,1]^k$.
    In particular, the best PQM is $\PQM_{L(\vec q)}$ and its SAR is given by formula \eqref{eq:ar_fixedq}.
\end{remark}

Given $k$, $n_u$, and $n_r$, we explicitly characterize the vector $\vec q$ that minimizes the SAR of $\PQM_{L(\vec q)}$.

\begin{theorem}
\label{thm:optvecqklenu}
    Given $k$, $n_u$, and $n_r$, let $\sigma,\lambda\in\enne$ be the unique pair of natural values such that $n_u=\lambda k+\sigma$.
    If $\RRn=\{\frac{2j-1}{2n_u}\}_{j\in[n_u]}$, the best PQM mechanism is $\PQM_{L(\vec q)}$ where $\vec q$ is 
    \[
        q_s=\frac{2(s-1)(\lambda+1)+\lambda+1}{2n_u}\; \text{if}\;\; s\le\sigma,\quad q_s=\frac{2(\sigma-1)(\lambda+1)+2(s-\sigma-1)\lambda+\lambda}{2n_u}\; \text{if}\;\; s>\sigma.
    \]
\end{theorem}

Owing to Theorem \ref{thm:optvecqklenu}, when $k$ divides $n_u$ and $n_r>n_u$, the best PQM is induced by the vector $\vec q=(\frac{1}{2k},\frac{3}{2k},\dots,\frac{2k-1}{2k})$ and its SAR is $1+\frac{2(1-\lambda)(\sigma-1)}{n_u-(1-\lambda)(\sigma-1)}$, where $\sigma\in\enne$ is such that $n_u=\sigma k$.
Moreover, notice that a similar argument allows us to handle the cases in which the cardinality of $\RRn$ is lower than $n_u$.
Lastly, we provide a lower bound on the SAR of any truthful mechanism that places the facility while having access to $k$ quantiles induced $\vec q=(\frac{1}{2k};\frac{3}{2k},\dots,\frac{2k-1}{2k})$.

\begin{theorem}
\label{thm:lbkwhatever}
    Let $k\in\enne$ be an integer such that there exists $\sigma\in\enne$ such that $n_u=\sigma k$.
    Given $\vec q\in [0,1]^k$ the vector containing $k$ equi-distanced values, that is $q_j=\frac{2j-1}{2k}$, any truthful mechanism $M$ that has access only to the quantiles induced by $\vec q$ is such that
    \[
        \SAR(M)\ge \begin{cases}
            1+\dfrac{2\frac{n_u}{k}}{n+n_u-2\frac{n_u}{k}}\quad &\text{if}\; k\,\text{is even}, \\  1+\dfrac{6\frac{n_u}{k}}{n+n_u-5\frac{n_u}{k}} &\text{otherwise.}
        \end{cases}
    \]
\end{theorem}

\begin{proof}[Sketch of the Proof]
    Let $M$ be a truthful mechanism.
    For the sake of argument, let us assume that $k$ and $n_r$ are even.
    By definition of $\vec q$, we have that $q_{\frac{k}{2}}=\frac{1}{2}-\frac{1}{2k}$ and $q_{\frac{k}{2}+1}=\frac{1}{2}+\frac{1}{2k}$.
    Let $\mu$ be a probability measure such that $F_\mu^{[-1]}(q_{\frac{k}{2}})=0$, $F_\mu^{[-1]}(q_{\frac{k}{2}+1})=1$, $F_\mu^{[-1]}(q_{j})\in(-\epsilon,0)$ if $j<\frac{k}{2}$, and $F_\mu^{[-1]}(q_{j})\in(1,1+\epsilon)$ if $j>\frac{k}{2}+1$, where $\epsilon$ is a small positive constant. 
    Let $\vec x\in\erre^{n_r}$ be a vector such that $x_i=0$ if $i\le \frac{n_r}{2}$ and $x_i=1$ otherwise.
    Finally, let us denote with $y$ the output of $M$ for this instance.
    Without loss of generality, we assume that $y\in[0,1]$ (as otherwise the SAR of the mechanism is higher).
    Unlike what we did in the proof of Theorem \ref{thm:PQMar}, we cannot restrict $y$ to be either $0$ or $1$, as the truthfulness of $M$ applies only to the agents' reports and not to the quantiles of $\mu$.
    Indeed, to maximize the ratio between the mechanism cost and the optimal cost, we either
    \begin{enumerate*}[label=(\roman*)]
        \item \label{case:AAA} move all the agents locate at $0$ to $y$ and concentrate all the probability that $\mu$ assigns to $(0,1)$ at $1$; or
        \item \label{case:BBB} move all the agents locate at $1$ to $y$ and concentrate all the probability that $\mu$ assigns to $(0,1)$ at $0$.
    \end{enumerate*}
    Whether we modify the instance following \ref{case:AAA} or \ref{case:BBB}, depends on which modification leads to the highest ratio.
    The lower bound is then retrieved by selecting the $y\in[0,1]$ that minimizes the maximum ratio attainable by either applying \ref{case:AAA} or \ref{case:BBB}.
    Owing to the symmetry of the instance, this happens when $y=\frac{1}{2}$.
    The full computation of the lower bounds is deferred to the Appendix.
\end{proof}

\section{Conclusion and Future Works}
\label{sec:conclusion}

In this paper, we introduce and study the Facility Location Problem with Aleatory Agents (FLPAA), where the facility has to accommodate agents whose position is known along with agents whose position is aleatory and described by a probability measure $\mu$.
After characterizing the optimal solution to the FLPAA for any given agent position $\vec x$ and any measure $\mu$, we studied the mechanism design aspects of the FLPAA.
We considered the problem of designing truthful mechanisms that perform well while having only access to partial information about the distribution $\mu$.
In particular, we assumed that the mechanism designer does not have access to $\mu$, but to $k$ quantiles that the mechanism designer can query for.
We introduced the notion of strong approximation ratio (SAR), which measures the ratio between the mechanism cost and the optimal cost on the worst-case input $\vec x$ and the worst-case distribution $\mu$.
We studied the upper and lower bounds for every possible value of $k$ and provided truthful routines with bounded SAR.
In several cases, the upper bound matches the lower bound.
Lastly, we extended our study to the case where we must locate two facilities with capacity $c$.
For future works, we aim to improve the lower bound for the case in which $1<k<n_u$.
It would also be interesting to study whether the i.i.d. assumption on the aleatory agents can be relaxed.
Another interesting development would be to study the FLPAA in a higher dimensional space and to study different costs other than the social cost, such as the Maximum Cost.

\bibliographystyle{alpha}
\bibliography{sample_wine}

\clearpage
\appendix

\section{Appendix}

In this appendix, we report all the material missing from the main body.

\section{Missing Proof}

In this appendix, we report all the proof missing from the main body of the paper.

\begin{proof}[Proof of Theorem \ref{thm:problem_opt}]
    It is well-known that the median of a c.d.f. is the element minimizing the absolute deviation.
    We then focus on the characterization of the optimal solution.
    First, we notice that since $\mu$ is absolutely continuous, we have that $F_{\lambda,\mu,\vec x}$ is discontinuous only at $x\in\mathcal{X}=\{x_i\}_{i\in[n_r]}$, where $x_i$ are the agents' reports.
    Thus $\mathcal{ESC}$ is continuous with respect to $y$.
    Moreover, it is differentiable on every point in $\erre\backslash\mathcal{X}$.
    From a simple computation, we infer that 
    \[
        \partial_y\mathcal{ESC}(y)=\lambda A(y)+(1-\lambda)B(y),
    \]
    where $A(y)=\frac{2\# \{x_i\le y\}-n_r}{n_r}$ and $B(y)=2F_\mu(y)-1$ for every $y\in\erre\backslash\mathcal{X}$.
    It is easy to see that both $A$ and $B$ are non decreasing in $y$.
    Moreover, $\lim_{y\to\infty}\partial_y\mathcal{ESC}(y)>0$ and $\lim_{y\to-\infty}\partial_y\mathcal{ESC}(y)<0$.
    We then define the sets $I_+$ and $I_-$ as $I_+=\{y\in \erre,\;\;\text{s.t.}\;\;\partial_y\mathcal{ESC}(y)>0 \}$ and $I_-=\{y\in \erre,\;\;\text{s.t.}\;\;\partial_y\mathcal{ESC}(y)<0 \}$.
    Finally, we set $\underline{y}=\sup\{I_+\}$ and $\overline{y}=\inf\{I_-\}$.
    Clearly, it holds $\underline{y}\le \overline{y}$.
    We now show that every $y$ such that $\underline{y}\le y\le  \overline{y}$ is a solution to Problem \ref{problem:fac1}.
    If $\underline{y}< \overline{y}$ it means that $\partial_y\mathcal{ESC}(y)=0$ for every $y\in(\underline{y},\overline{y})$, thus every $y$ is a minimizer.
    Since $\mathcal{ESC}$ is continuous, we conclude that every $y\in [\underline{y},\overline{y}]$ is a minimizer.
    Finally, if $\underline{y} = \overline{y}=y^*$, we have that $\mathcal{ESC}$ is decreasing for every $y\le y^*$ and increasing for every $y\ge y^*$, which concludes the proof.
    Lastly, if none of the agents' positions $x_i$ are an optimal solution, it means that $\Flam(x_i)\neq \frac{1}{2}$ for every $i\in[n_r]$.
    It means that $F_\mu(y)>0$ for every $y$ that is optimal, hence $y$ belongs to the support of $\mu$.
    The uniqueness follows from the fact that $F_\mu$ is bijective.
\end{proof}

\begin{proof}[Proof of Theorem \ref{thm:setoffeasibleopt}]
    Let $\vec x$ be the vector containing the reports of the agents.
    If there exists $x_i$ such that $y^*=x_i$ is a solution to Problem \ref{problem:fac1}, we have nothing to prove.
    Let us then consider the case in which none of the agents reports is optimal.
    Let us denote with $y$ the optimal solution.
    By Theorem \ref{thm:problem_opt}, we have that $y\in spt(\mu)$ and that
    \begin{equation}
        \label{eq:card_discreteproof}
        \frac{\#\{x_i\le y\}}{n}+\frac{n_u}{n}F_\mu(y)=\frac{1}{2}.
    \end{equation}
    Denoted with $k$ the cardinality of the set $\{x_i\le y\}$, we infer that $y$ must satisfy the following identity 
    \[
        F_\mu(y)=\frac{n-2k}{2n_u}=\frac{2(\frac{n+1}{2}-k)-1}{2n_u},
    \]
    where $k\le\frac{n-1}{2}$.
    The proof follows, by adopting the change of variable $s=(\frac{n+1}{2}-k)$.
\end{proof}

\begin{proof}[Proof of Corollary \ref{cor:essentialpoints}]
    It follows from the argument used to prove Theorem \ref{thm:setoffeasibleopt}.
    Indeed, since the cardinality of the set $\{x_i\le y\}$ in \eqref{eq:card_discreteproof} is at most $n_r$, it means that $F_\mu(y)$ has to be larger than $\frac{n-2n_r}{2n_u}$.
    By enforcing this further restriction to the set $\{F^{[-1]}_\mu(\frac{2s-1}{2n_u})\}_{s\in[n_u]}$, we conclude the proof.
\end{proof}

\begin{proof}[Proof of Theorem \ref{them:truthfulPQM}]
    It follows from the fact that, for every given $\mu$, the routine of the $\PQM_{\vec q}$ is the same as the routine of a Phantom Peak Mechanism, which is truthful and anonymous \cite{moulin1980strategy}.
\end{proof}

\begin{proof}[Proof of Theorem \ref{thm:SAR0inf}]
Let $\vec x\in\erre^{n_r}$ be the vector containing the agents' reports.
Without loss of generality, we assume that $x_i\neq x_j$ for every couple of indexes $i\neq j$ and that $\vec x$ is ordered increasingly, that is $x_i<x_j$.
Furthermore, we denote with $y$ the output of the median mechanism, so that $y=x_{\floor{\frac{n_r+1}{2}}}$.
Finally, given $\mu$, let us denote with $y^*\in\erre$ the element that minimizes \eqref{eq:defproblemonefacility}. 
Without loss of generality, assume that $y^*=0$; moreover, since the other case is symmetric, we assume that $y<y^*=0$.
To compute the SAR of the median mechanism, we first provide an upper bound on the SAR (Step 1) and then build a sequence of instances such that the ratio of the ESC attained by the mechanism and the optimal ESC converges to the upper bound we obtained (Step 2).

\textbf{Step 1.} 
First, we provide an upper bound on the the ESC attainable by the median mechanism.
Let us fix a distribution $\mu$, owing to the triangular inequality and recalling that $y=M(\vec x)$, we have that
\begin{align}
    \nonumber\ESC_M(\vec x,\mu)&=\sum_{i=1}^{n_r}|x_i-y|+n_u\EE_{X\sim \mu}[|X-y|]\le \sum_{i=1}^{n_r}|x_i-y|+n_u\Big(\EE_{X\sim \mu}[|X|]+\EE_{X\sim \mu}[|y|]\Big)\\
    &=\sum_{i=1}^{n_r}|x_i-y|+n_u\Big(\EE_{X\sim \mu}[|X|]+|y|\Big).
\end{align}
So that, since $y^*=0$, we have that
\begin{equation}
\label{eq:SAR:ratio0inf}
    \frac{\ESC_M(\vec x ;\mu)}{\ESC_{opt}(\vec x;\mu)}\le \frac{\sum_{i=1}^{n_r}|x_i-y|+n_u\Big(\EE_{X\sim \mu}[|X|]+|y|\Big)}{\sum_{i=1}^{n_r}|x_i|+n_u\EE_{X\sim \mu}[|X|]}\le \frac{\sum_{i=1}^{n_r}|x_i-y|+n_u|y|}{\sum_{i=1}^{n_r}|x_i|}.
\end{equation}
Lastly, we upper bound the ratio in \eqref{eq:SAR:ratio0inf} by setting $x_i=y$ if $i\le \floor{\frac{n_r+1}{2}}$ and $x_i=0$ otherwise.
Notice that $M(y,y,\dots,y,0,\dots,0)=y$ by definition, moreover we have that
\begin{enumerate}
    \item changing the position of every agent located between $y$ and $y^*=0$ to $0$ increases the cost of the mechanism and decreases the optimal cost, resulting in a higher ratio in \eqref{eq:SAR:ratio0inf}.
    \item Changing the position of every agent $x_i\le y$ to $y$, we decreases both the numerator and the denominator of \eqref{eq:SAR:ratio0inf} by the same quantity, that is $\sum_{i=1}^{\floor{\frac{n_r-1}{2}}}|x_i-y|$, thus increasing the ratio.
    \item Similarly, moving the position of every agent $x_i\ge y^*=0$ to $0$ increases the ratio in \eqref{eq:SAR:ratio0inf}. 
\end{enumerate}
Moreover, owing to Theorem \ref{thm:problem_opt} and to the definition of $M$, neither the optimal location of the facility nor the facility position returned by the mechanism change position, which allows us to conclude that
\begin{equation}
\label{eq:upperboundproofSAR0info}
    \frac{\ESC_M(\vec x;\mu)}{\ESC_{opt}(\vec x;\mu)}\le\frac{\floor{\frac{n_r-1}{2}}|y|+n_u|y|}{\floor{\frac{n_r+1}{2}}|y|}\le\frac{\floor{\frac{n_r-1}{2}}+n_u}{\floor{\frac{n_r+1}{2}}}.
\end{equation}
Since the latter bound does not depend on either $\vec x$ nor $\mu$, we infer $SAR(M)\le\frac{\floor{\frac{n_r-1}{2}}+n_u}{\floor{\frac{n_r+1}{2}}}$.
\textbf{Step 2.}
We now conclude the proof by building a sequence of instance of the FLPAA such that the ratio in \eqref{eq:SAR:ratio0inf} converges to the upper bound \eqref{eq:upperboundproofSAR0info}.
The sequence of instances is composed by a vector containing the agents reports and an probability measure.
Let us denote with $\vec x$ the instance in which $\floor{\frac{n_r+1}{2}}$ agents are located at $y$ and the remaining agents are located at $0$, to these agents reports, we pair the measure $\mu_\ell$.
We now build the sequence of probability measures $\mu_\ell$.
Given $\ell\in\enne$ and given $\mu$ the probability measure from Step 1, we  define $\mu_\ell$ as the probability measure induced by the density $\ell\rho_\mu(\ell\,x)$, where $\rho_\mu$ is the density function associated with $\mu$.
Owing again to Theorem \ref{thm:problem_opt}, for every $\ell\in\enne$, $y^*=0$ remains an optimal facility location.
From a simple computation, we infer that
\begin{equation}
    \EE_{X\sim\mu_\ell}[|X|]=\int_\erre |x|\ell\rho_\mu(\ell\,x)dx=\int_\erre\Big|\frac{t}{\ell}\Big|\rho_\mu(t)dt=\frac{1}{\ell}\int_\erre|t|\rho_\mu(t)dt,
\end{equation}
where the second equality follows by applying the change of variable $t=\ell\,x$.
Through a similar argument, we have that
\begin{equation}
    \EE_{X\sim\mu_\ell}[|X-y|]=\int_\erre |x-y|\ell\rho_\mu(\ell\,x)dx=\int_\erre\Big|\frac{t}{\ell}-y\Big|\rho_\mu(t)dt.
\end{equation}
Since $\int_\erre|x|\rho_\mu(x)dx<+\infty$, we can apply Lebesgue's theorem to conclude that
\begin{equation}
    \lim_{\ell\to\infty}\EE_{X\sim\mu_\ell}[|X|]=0 \quad\quad\quad \text{and} \quad\quad\quad \lim_{\ell\to\infty}=\EE_{X\sim\mu_\ell}[|X-y|]=|y|.
\end{equation}
We then conclude that
\begin{equation}
    \lim_{\ell\to\infty}\frac{\ESC_M(\vec x;\mu_\ell)}{\ESC_{opt}(\vec x;\mu_\ell)}=\frac{\floor{\frac{n_r-1}{2}}|y|+n_u\EE_{X\sim\mu_\ell}[|X-y|]}{\floor{\frac{n_r+1}{2}}|y|+n_u\EE_{X\sim\mu_\ell}[|X|]}=\frac{\floor{\frac{n_r-1}{2}}+n_u}{\floor{\frac{n_r+1}{2}}},
\end{equation}
which allows us to conclude the proof.
\end{proof}

\begin{proof}[Proof of Theorem \ref{thm:app}]
    Let us consider the instance $\vec x$ defined as $x_i=-1$ if $i\le \frac{n_r+1}{2}$ and $x_i=0$ otherwise.
    To prove our lower bound, we consider a sequence of uniform probability distributions.
    Let us denote with $M$ a truthful mechanism and with $y$ the output of $M$ given $\vec x$ in input.
    Without loss of generality, we assume that $y\in\{-1,0\}$.
    Indeed, if $y\notin\{-1,0\}$, we can use the truthfulness of $M$ to move either all the agents located at $-1$ or all the agents located at $0$ to $y$.
    Thus, up to a scale change, the argument we present next would hold in the same way. 

    \textbf{Case $y=-1$.}
    For every $\ell\in\enne$, we denote with $\mu_\ell$ the uniform probability distribution over the set $[-\frac{1}{2\ell},\frac{1}{2\ell}]$, so that $\mu_\ell=\mathcal{U}_{[-\frac{1}{2\ell},\frac{1}{2\ell}]}$.
    Notice that for every $\ell$, the optimal solution to the problem, namely $y_\ell^*$ belongs to the support of $\mu_\ell$, thus $y_\ell^*\in[-\frac{1}{2\ell},\frac{1}{2\ell}]$.
    In particular, we infer that the optimal solution converges to $y^*=0$ as $\ell$ approaches infinity.
    Moreover, we have that
    \begin{equation}
        \EE_{X\sim\mu_\ell}[|X-y_\ell^*|]\le \EE_{X\sim\mu_\ell}[|X|] + \EE_{X\sim\mu_\ell}[|y_\ell^*|]\le \EE_{X\sim\mu_\ell}[|X|] +\frac{1}{\ell},
    \end{equation}
    therefore $\lim_{X\sim\mu_\ell}\EE_{X\sim\mu_\ell}[|X-y_\ell^*|]=0$.
    We then conclude that
    \begin{equation}
        \lim_{\ell\to\infty}\frac{\ESC_M(\vec x;\mu_\ell)}{\ESC_{opt}(\vec x;\mu_\ell)}=\frac{\floor{\frac{n_r-1}{2}}|y|+n_u\EE_{X\sim\mu_\ell}[|X+1|]}{\floor{\frac{n_r+1}{2}}|y|+n_u\EE_{X\sim\mu_\ell}[|X-y^*_{\ell}|]}=\frac{\floor{\frac{n_r-1}{2}}+n_u}{\floor{\frac{n_r+1}{2}}},
    \end{equation}
    which concludes the proof for this case.

    \textbf{Case $y=0$.} 
    In this case, we consider the same sequence $\mu_\ell$ used to handle the case in which $y=-1$ but translated to include $-1$, that is $\mu_\ell=\mathcal{U}_{[-1-\frac{1}{2\ell},-1+\frac{1}{2\ell}]}$.
    By repeating the argument used in the previous case, we obtain a higher ratio than the one obtained in Case 1, which concludes the proof.
\end{proof}

\begin{proof}[Proof of Theorem \ref{thm:opt_discretemech}]
    The truthfulness of $\PQM_{\vec q}$ follows from Theorem \ref{them:truthfulPQM}.
    We now prove that the output of $\PQM_{\vec q}$ is always an optimal solution.
    We have two cases: either $\PQM_{\vec q}(\vec x)=f_s$ or $\PQM_{\vec q}(\vec x)=x_l$ for some $s\in[n_r]$ or $l\in[n_u]$, respectively.
    For the sake of simplicity, we assume that $n=n_u+n_r$ is odd and that no agent reports the position of the quantile and couple of agents report the same value, i.e. $x_i\neq x_j$ for every $i\neq j$.
    The proof can easily be extended to the case in which $n$ is even and/or two or more agents report the same position.
    {\textbf{Case 1: $\PQM_{\vec q}(\vec x)=f_s$}.}
    Let us set $s^c=\frac{n+1}{2}-s$.
    By definition of the mechanism, we have that $|\{x_i\le f_s\}|=s^c$, then we have that
    \begin{align*}
        \partial_y \mathcal{ESC}(\vec x; f_s, \mu)&=2\frac{s^c}{n}+(1-\lambda)(2F_\mu(f_s))-1=2\frac{s^c}{n}+(1-\lambda)\frac{2s-1}{n_u}-1\\
        &=\frac{2s^c}{n}+\frac{2s-1}{n}-1=\frac{n+1-2s}{n}+\frac{2s-1}{n}-1=0,
    \end{align*}        
    where we used the identity $1-\lambda=\frac{n_u}{n}$.
    Since $\partial_y\mathcal{ESC}(\vec x; f_s, \mu)=0$, we have that $f_s$ the optimal solution to the problem.
    {\textbf{Case 2: $\PQM_{\vec q}(\vec x)=x_l$}.}
    First, let us assume that there exists $s\in [n_u-1]$ such that $f_s\le x_l\le f_{s+1}$.
    By definition of the mechanism, we have that $s+l=\frac{n+1}{2}$, that is $s=\frac{n+1}{2}-l$.
    In this case, we have that
    \begin{align*}
        \partial_y \mathcal{ESC}(\vec x ;x_l,\mu)&=\frac{2l}{n}+(1-\lambda)(2F_\mu(x_l))\ge \frac{2l}{n}+(1-\lambda)\frac{n+1-2l-1}{n_u}-1\\
                                    &=\frac{2l}{n}+\frac{n+1-2l-1}{n}-1=0.
    \end{align*}
    Let us now consider $\epsilon>0$ such that $x_{l-1}<x_{l}-\epsilon$ and $f_s<x_{l}-\epsilon$, then we have
    \begin{align*}
        \partial_y \mathcal{ESC}(\vec x, x_l-\epsilon;\mu)&=\frac{2l-2}{n}+(1-\lambda)(2F_\mu(x_l-\epsilon))\le \frac{2l-2}{n}+(1-\lambda)\frac{n-2l+2}{n_u}-1\\
                                    &=\frac{2l-2}{n}+\frac{n-2l+2}{n}-1=0.
    \end{align*}
    Thus, from Theorem \ref{thm:problem_opt}, we infer that $x_s$ is optimal, which concludes the proof.
    A similar argument applies to the case in which $x_l<f_1$ or $x_l>f_{n_u}$.
\end{proof}

\begin{proof}[Proof of Theorem \ref{thm:armedianmech}]
    From Theorem \ref{thm:setoffeasibleopt} and \ref{thm:opt_discretemech}, we have that $\PQM_{\vec q}$ is optimal for $n_u\in\{0,1,n\}$.
    % 
    % In the following, we consider in details the proof for the case in which $\lambda\ge \frac{1}{2}$, the proof for the case $\lambda<\frac{1}{2}$ follows a similar argument and thus is just sketched.
    % % 
    % For a full version of the proof for the case $\lambda<\frac{1}{2}$, we refer the reader to the Appendix.
    % 
    % Moreover, 
    Without loss of generality, we assume that the worst-case instance occurs when there are more agents reports on the left of $m$, so that the optimal facility position $y^*$ is on the left of $m$, hence $y^*<m$.
    Given $t\in\erre$ such that $t<m$, let us consider the following instance $x_1=\dots=x_{\frac{n-1}{2}}=t$ and $x_j=m$ for all the other $j\in[n_r]$, so that the mechanism places the facility at $m$.
    We denote with $\vec x$ the vector containing the agents' reports.
    Since $\lambda\ge\frac{1}{2}-\frac{1}{2n}$, this instance is well-defined.
    Without loss of generality, let us now consider a measure $\mu$ whose median is $m$ and such that $F_\mu(\frac{1}{2n_u})=t$.
    Indeed, if $t\neq F_\mu(\frac{1}{2n_u})<m$, as otherwise we can move the agents to $F_\mu(\frac{1}{2n_u})$ without altering the outcome of the mechanism.
    By Theorem \ref{thm:opt_discretemech}, if the aleatory agents are distributed according to $\mu$, the optimal position of the facility for instance $\vec x$ is $t$.
    By definition, $\PQM_{\vec q}$ places the facility at $m$.
    We now individually consider the cost of the mechanism and the optimal cost.
    We first start from the optimal cost.
    We have the following
    \[
        \mathcal{ESC}_{opt}(\vec x;\mu)=\Big(n_r-\frac{n-1}{2}\Big)|t-m|+n_u\EE_{X\sim\mu}[|X-t|].
    \]
    By definition of expected value, we have that
    \begin{align*}
        \EE_{X\sim\mu}[|X-t|]&=\int_{-\infty}^t(t-x)d\mu+\int_t^{+\infty}(x-t)d\mu\\
        &=\int_{-\infty}^t(t-x)d\mu+\int_t^{m}(x-t)d\mu+\int_m^{+\infty}(x-t)d\mu\\
        &=\int_{-\infty}^t(t-x)d\mu+\int_t^{m}(x-t)d\mu+\int_m^{+\infty}(x-m)d\mu+\frac{(m-t)}{2}\\
        &\ge \int_{-\infty}^t(t-x)d\mu+\int_m^{+\infty}(x-m)d\mu+\frac{(m-t)}{2}
    \end{align*}
    since $m$ is the median of $\mu$, thus $1-F_\mu(m)=\frac{1}{2}$.
    Similarly, we handle the mechanism cost, which is
    \[
        \mathcal{ESC}_{\PQM_{\vec q}}(\vec x;\mu)=\frac{n-1}{2}|t-m|+n_u\EE_{X\sim\mu}[|X-m|].
    \]
    Again, due to the definition of expected value, we have that
    \begin{align*}
        \EE_{X\sim\mu}[|X-m|]&=\int_{-\infty}^t(t-x)d\mu+\int_t^{m}(m-x)d\mu+\int_m^{+\infty}(x-m)d\mu+\frac{(m-t)}{2n_u}\\
        &\le \int_{-\infty}^t(t-x)d\mu+\int_m^{+\infty}(x-m)d\mu+\frac{(m-t)}{2}.
    \end{align*}
    Since $(m-x)\le (m-t)$ when $x\in[t,m]$, hence $\int_t^{m}(m-x)d\mu\le (m-t)(\frac{1}{2}-\frac{1}{2n_u})$.
    By combining these two estimations, we infer
    \begin{align*}
        ar_\mu(Med_{\vec m})&\le \frac{\frac{n-1}{2}+n_u(\int_{-\infty}^t\frac{t-x}{m-t}d\mu+\int_m^{+\infty}\frac{x-m}{m-t}d\mu+\frac{1}{2})}{n_r-\frac{n-1}{2}+n_u(\int_{-\infty}^t\frac{t-x}{m-t}d\mu+\int_m^{+\infty}\frac{x-m}{m-t}d\mu+\frac{1}{2})} \\
        &\le\frac{\frac{n-1}{2}+\frac{n_u}{2}}{n_r-\frac{n-1}{2}+\frac{n_u}{2}}=\frac{\frac{n-1}{2}+\frac{n_u}{2}}{n_r-\frac{n-1}{2}+\frac{n_u}{2}}
        =\frac{(2-\lambda)n-1}{\lambda n +1}=\frac{2}{\lambda+\frac{1}{n}}-1,
    \end{align*}
    for every $\mu\in\PP(\erre)$.
    Thus, we have $SAR(M)\le\frac{2}{\lambda+\frac{1}{n}}-1$.
    To conclude the proof, we show that there exists a sequence of probability distributions $\mu_\ell$ such that
    $\lim_{\ell\to\infty}\frac{\ESC_{\PQM_{\vec q}(\vec x;\mu_\ell)}}{\ESC_{opt}(\vec x;\mu_\ell)}=\frac{2}{\lambda+\frac{1}{n}}-1$.
    Let us consider the instance $\vec x$ in which $x_1=\dots=x_{\floor{\frac{n-1}{2}}}=0$ and $x_i=1$ otherwise.
    We then define $\mu_{\ell}=\frac{\ell}{2n_u}\mathcal{U}_{[-\frac{1}{2\ell},0]}+\frac{\ell(2n_u-1)}{2n_u}\mathcal{U}_{[1-\frac{1}{2\ell},1]}$.
    By repeating the argument of Step 2 in Theorem \ref{thm:SAR0inf}, we conclude the proof. 
\end{proof}

\begin{proof}[Proof of Theorem \ref{thm:LBmedian}]
    Let $m$ denote the median of $\mu$.
    We first study the case $\lambda\le \frac{1}{3}$ and then the case $\lambda\ge \frac{1}{3}$.
    Let $\lambda\in[0,\frac{1}{3}]$.
    We consider the instance $x_1=\dots=x_{n_r}=t$, where $t\in\erre$ is such that $t<m$.
    Let us denote with $y$ the position returned by the mechanism on this instance.
    Without loss of generality, we assume that $y\in[t,m]$.
    Indeed, if $y\notin [t,m]$ the approximation ratio of the mechanism would be higher than what it achieves by placing $y\in[t,m]$.
    Moreover, due to the truthfulness of the mechanism, we can assume that the output of a mechanism is either $m$ or $t$.
    Indeed, if $y\neq t$, then the mechanism places the facility at $t$ if all the agents placed at $t$ move to $y$.
    If the mechanism places the facility at $t$ we consider a sequence of measures $\mu_\ell$ that concentrates all the mass at $m$, such as $\mu_\ell=2\ell\,\mathcal{U}_{[m-\frac{1}{\ell},m+\frac{1}{\ell}]}$.
    In this case, the limit of the ratios between the mechanism cost and the optimal cost is
    \[
        ar_\mu(M)=\frac{n_u}{n_r}=\frac{1-\lambda}{\lambda}=1+\frac{1-2\lambda}{\lambda}.
    \]
    If the mechanism places the facility at $m$, we consider a sequence of probability distribution $\mu_\ell$ defined as $\mu_\ell=\ell\mathcal{U}_{[t-\frac{1}{\ell},t+\frac{1}{\ell}]}+ell\mathcal{U}_{[m-\frac{1}{\ell},m+\frac{1}{\ell}]}$.
    In this case, the optimal location for the facility is $t$, thus the limit of the ratios between the mechanism cost and the optimal cost is
    \[
        ar_\mu(M)=\frac{(n_r+\frac{n_u}{2})|t-m|}{\frac{n_u}{2}|m-t|}=1+2\frac{\lambda}{1-\lambda}.
    \]
    We notice that $2\frac{\lambda}{1-\lambda} \le \frac{1-2\lambda}{\lambda}$ whenever $\lambda\in[0,\frac{1}{3})$, thus we infer that
    \[
        ar_\mu(M)\ge 1+2\frac{\lambda}{1-\lambda},
    \]
    that is $SAR(M)\ge 1+2\frac{\lambda}{1-\lambda}$, which concludes the proof.
    Let us now consider the case in which $\lambda>\frac{1}{3}$.
    Notice that owing to Theorem \ref{thm:opt_discretemech}, if $n_r=n,n-1$, there is an optimal mechanism, so we tacitly assume that $n_r\le n-2$.
    To prove this lower bound, we need to consider an instance in which part of the agents are located at $t<m$ and the remaining agents are placed at $m$.
    We denote with $l$ (as for \textit{left agents}) the number of agents placed at $t$, so that $r=n_r-l$ (as for \textit{right agents}) is the amount of agents placed at $m$.
    Since we want that the optimal solution of the problem lies between $t$ and $m$, we assume that $l>r$.
    Owing to the truthfulness of the mechanism, we need to consider only the cases in which the mechanism places the facility at $t$ or $m$.
    First, we consider the case in which $y=t$.
    In this case, we consider a sequence of probability distributions $\mu_\ell$ that concentrate all the probability at $m$, as, for example $\mu_\ell=2\ell\mathcal{U}_{[m-\frac{1}{\ell},m+\frac{1}{\ell}]}$.
    The median of each $\mu_\ell$ is $m$, so the mechanism does not change its output.
    By the same argument adopted in Step 2 during the proof of Theorem \ref{thm:SAR0inf}, we have that 
    \[
    \lim_{\ell\to\infty}ar_{\mu_\ell}(M)\ge  \frac{n-l}{l}=\frac{n}{l}-1.
    \]
    We then study the case in which $y=m$.
    In this case, we want to define a sequence of probability distributions $\mu_\ell$ that concentrates half probability to $t$ and half at $m$, so that for every $\ell\in\enne$, the median of $\mu_\ell$ remains $m$.
    Consider, for example, $\mu_{\ell}=\frac{\ell}{2}\mathcal{U}_{[-\frac{1}{\ell},0]}+\frac{\ell}{2}\mathcal{U}_{[1-\frac{1}{\ell},1]}$.
    Again, by taking the limit for $\ell\to\infty$, the ratio of the mechanism cost and the optimal cost converges to 
    % $\frac{k}{n}|t-m|+\frac{n_u}{2n}|t-m|$.
    % % 
    % Since $k>l$, we have that the optimal cost is
    % \[
    % \frac{l}{n}|t-m|+\frac{n_u}{2n}|t-m|.
    % \]
    % Hence we have that the approximation ratio is
    \[
    % ar_\mu(M)\ge 
    \frac{2l+n_u}{2(n_r-l)+n_u}=\frac{2n}{2n-(2l+n_u)}-1.
    \]
    For every $l$ we then have
    \[
    SAR(M)\ge \min\Big\{\frac{n}{l},\frac{2n}{2n-(2l+n_u)}\Big\}-1
    \]

    We now tune $l$ in order to make the bound on $SAR(M)$ as large as possible.
    To do that, we impose $\frac{n}{l}=\frac{2n}{2n-(2l+n_u)}$, that is $l=n-l-\frac{n_u}{2}$, thus $l=\frac{2n-n_u}{4}$.
    Notice that $l$ has to be an integer, thus we consider $l=\ceil{\frac{2n-n_u}{4}}$, from which we infer that $ar_\mu(M)\ge \frac{n}{\ceil{\frac{2n-n_u}{4}}}-1$ for every $\mu$, thus we infer $SAR(M)\ge \frac{n}{\ceil{\frac{2n-n_u}{4}}}-1$.
    If $\frac{2n-n_u}{4}\in\enne$, we have that
    \[
        \frac{4n}{2n-n_u}-1= \frac{4}{2-(1-\lambda)}-1=\frac{4}{1+\lambda}-1,
    \]
    which concludes the proof.
\end{proof}

\begin{proof}[Proof of Theorem \ref{thm:PQMar}]
    First, we keep $\vec q\in[0,1]^k$ fixed and study the SAR of $\PQM_{L(\vec q)}$.
    By the same argument used to prove Theorem \ref{thm:armedianmech}, we have that the worst-case instance occurs when the optimal position of the facility is $f_{k}\in\RRn$, while the mechanism returns a different position namely $f_{k'}$.
    We denote with $q_k$ and $q_{k'}$ the values for which it holds $f_k=F_{\mu}^{[-1]}(q_k)$ and $f_{k'}=F_{\mu}^{[-1]}(q_{k'})$, respectively.
    Let us set $\Delta_q=|q_k-q_{k'}|$.
    Since the other case is symmetric, we assume that $f_{k'}<f_{k}$, hence $q_{k'}<q_{k}$.
    Let us then consider $\vec x$ such that the optimal location for the facility is $f_k\in\RRn$ and $\PQM_{L(\vec q)}(\vec x)=f_{k'}$.
    We recall that the optimal location $f_k$ and $\PQM_{L(\vec q)}(\vec x)=f_{k'}$ are the median of $(\vec x,\vec f)$\footnote{We recall that $\vec f=(f_1,\dots,f_{n_u})$, where $f_j=F_\mu^{[-1]}(\frac{2j-1}{2n_u})$} and $(\vec x, L(\vec q))$, respectively.
    In particular, we can assume that if $x_i\le f_k$ then $x_i\le f_{k'}$ as well, as otherwise we can move any agent whose position $x_i\in(f_{k'},f_k)$ to $f_k$ without altering the output of $\PQM$ and the optimal position and while increasing the ratio between the mechanism and optimal costs.
    Likewise, the ratio increases if we move all the agents on the left of $f_{k'}$ to $f_{k'}$ and all the agents on the right of $f_k$ to $f_k$.
    Since $f_k$ is optimal, we must have that $\Flam(f_k)\ge \frac{1}{2}$, while, by definition of $f_{k'}$, we have that $\Flam(f_{k'})\le \frac{1}{2}-(1-\lambda)\Delta_q$.
    By the same argument used in the proof of Theorem \ref{thm:armedianmech}, we have that the ratio between the mechanism cost and the optimal cost increases when the measure $\mu$ concentrate as much probability as possible at $f_{k}$ and the remaining at $f_{k'}$.
    We can then build a sequence of probability measures $\mu_\ell$ such that
    \begin{equation}
        \label{eq:optESCproofarkquant}
        \lim_{\ell\to\infty}\ESC_{opt}(\vec x;\mu_\ell)=\Big(\frac{1}{2}-(1-\lambda)\Delta_q\Big)|f_k-f_{k'}|
    \end{equation}
    and
    \begin{equation}
        \label{eq:mechESCproofarkquant}
        \lim_{\ell\to\infty}\ESC_{\PQM}(\vec x;\mu_\ell)=\Big(\frac{1}{2}+(1-\lambda)\Delta_q\Big)|f_k-f_{k'}|
    \end{equation}
    thus the ratio is equal to
    \begin{equation}
    \label{eq:rationonsup}
        \frac{\frac{1}{2}+(1-\lambda)\Delta_q}{\frac{1}{2}-(1-\lambda)\Delta_q}=1+\frac{4(1-\lambda)\Delta_q}{1-2(1-\lambda)\Delta_q}.
    \end{equation}
    Lastly, we notice that the function \eqref{eq:rationonsup} is increasing with respect to $\Delta_q$, thus the approximation ratio is equal to the maximum amongst the possible values of $\Delta_q$, that is $\Dnn(L(\vec q))$, which concludes the first part of the proof.
    To conclude we show that $\PQM_{L(\vec q)}$ is the optimal PQM given $\vec q$.
    Since the approximation ratio of $\PQM_{\vec w}$ (see \eqref{eq:ar_fixedq}) is increasing with respect to $\Dnn(\vec w)$, the mechanism with the lowest approximation ratio is induced by the vector with the smallest $\Dnn(\vec w)$.
    Given $\vec q$, we have that, by definition of $L$, the vector that minimizes $\Dnn$ is $L(\vec q)$, which concludes the proof.
\end{proof}

\begin{proof}[Proof of Theorem \ref{thm:optvecqklenu}]
    By definition of $\Dnn$ and $L$, the optimal $\vec q\in[0,1]^k$ is a solution to the following problem
    \begin{equation}
        \min_{\vec q\in [0,1]^k}\max_{j\in\RRn}\min_{i\in[k]}\Big|q_i-\frac{2j-1}{n_u}\Big|.
    \end{equation}
    Notice that any $\vec q\in [0,1]$ divides the set $\RRn$ into $k$ sets, namely $A_i$ where $i\in[k]$, where $A_i=\{j\in\RRn\quad\text{s.t.}\quad|q_i-\frac{2j-1}{2n_u}|=\min_{i\in[k]}|q_i-\frac{2j-1}{2n_u}|\}$.
    Notice that each $A_i\cap A_l=\emptyset$ for every $i\neq l$, while $\cup_{i\in[k]}A_i=\RRn$.
    Then, the $\vec q$ that minimizes $\Dnn(L(\vec q))$ and induces the same partition $A_i$ is $q_i=mid(A_i)$, where $mid(A_i)$ is the middle point of $A_i$, moreover, we have $\max_{j\in\RRn}\min_{i\in[k]}|q_i-\frac{2j-1}{n_u}|=\max_{i\in[k]}|\max(A_i)-mid(A_i)|$.
    Since the elements in $\RRn$ are equi-distanciated, the partition $A_i$ that minimizes $\max_{i\in[k]}|\max(A_i)-mid(A_i)|$ divides the set $\RRn$ into $k$ sets as evenly as possible.
    In particular, if $R=sk+l$ with $s,l\in\enne$ and $l<R$, the optimal partition is composed of $l$ sets with $s+1$ elements and $R-l$ sets with $s$ elements.
    Lastly, if $\RRn=\{\frac{2j-1}{2n_u}\}_{j\in[n_u]}$, we have that $\hat{A}_i=\{\frac{2j-1}{2n_u}\}_{(s+1)(i-1)+1\le j\le (s+1)i}$ if $i\le l$ and $\hat{A}_i=\{\frac{2j-1}{2n_u}\}_{l(s+1)+s(i-s-1)+1\le j\le (s+1)i+s(i-s)}$ for all the other $i\in[k]$ is an optimal partition of $\RRn$.
    The formula of the optimal percentile vector $\vec q$ is obtained by computing the middle point of each $\hat{A}_i$.
\end{proof}

\begin{proof}[Proof of Theorem \ref{thm:lbkwhatever}]
Let $k$ be the number of equi-distanced quantiles the mechanism designer has access to.
We divide our study into two cases, depending on whether $k$ is even or odd.
First, we consider the case in which $k$ is even.
From our hypothesis, we infer that $q_{\frac{k}{2}}=\frac{1}{2}-\frac{1}{2k}$ and $q_{\frac{k}{2}+1}=\frac{1}{2}+\frac{1}{2k}$.
In particular, $|q_{\frac{k}{2}}-q_{\frac{k}{2}+1}|=\frac{1}{k}$.
Without loss of generality, we assume that $F_\mu^{[-1]}(q_{\frac{k}{2}})=0$ and $F_\mu^{[-1]}(q_{\frac{k}{2}+1})=1$.
Moreover, $F_\mu^{[-1]}(q_i)\in(-\epsilon,0)$ if $i\le\frac{k}{2}$ and $F_\mu^{[-1]}(q_i)\in(1,1+\epsilon)$ for every other $i$, where $\epsilon>0$ is a arbitrarily small parameter. \footnote{This assumption is only necessary to stick with the basic assumptions described in the main body of the paper. The reader can assume that $F_\mu^{[-1]}(q_1)=\dots=F_\mu^{[-1]}(q_{\frac{k}{2}})$ and $F_\mu^{[-1]}(q_{\frac{k}{2}})=\dots=F_\mu^{[-1]}(q_k)$ and the proof would still hold.}
Let us consider $\vec x$ in which the agents are split evenly between $0$ and $1$, that is $x_1=\dots=x_{\ceil{\frac{n_r}{2}}}=0$ and $x_i=1$ otherwise.
Let $y$ be the position at which a truthful mechanism $M$ places the facility on this instance.
Without loss of generality, we assume that $y\in[0,1]$.
It is worthy of notice that we cannot restrict to the case $y=0$ or $y=1$, as altering the position of the quantiles alters the outcome of the mechanism $M$.
For the sake of argument, let us consider $n_r$ to be even, if $n_r$ is odd the proof follows by a similar argument.
If $n_r$ is even, there are the same amount of reporting agents at $0$ and at $1$.
Depending on the position of $y$, we use the truthfulness of $M$ to move all the agents at $0$ to $y$ or all the agents at $1$ to $y$.
If we move all the agents from $0$ to $y$, we have to consider a sequence of probability measures $\mu_\ell$ that assigns all the probability on $(0,1)$ (which we recall is $\frac{1}{k}$) to $1$, so that the optimal facility position is $1$.
Vice-versa, if we move all the agents at $1$ to $y$, we consider a sequence that concentrate all the probability at $0$.
In both cases, we compute the ratio between the mechanism cost and the optimal cost and decide whether to move the agents from $0$ to $y$ or the agents from $1$ to $y$ depending on which actions leads to the highest ratio.
To determine the lower bound, we then need to select the $y$ that minimizes the maximum ratio achievable in this way.
Due to the symmetry of the instance, we have that the best possible position at which the mechanism can place the facility is $y=\frac{1}{2}$.
In which case, we have that
\[
    \lim_{\ell\to\infty}\frac{\ESC_M(\vec x;\mu_\ell)}{\ESC_{opt}(\vec x;\mu_\ell)}=\dfrac{\frac{n+n_u}{2}}{\frac{n+n_u}{2}-\frac{n_u}{k}}=1+\frac{2\sigma}{n+n_u-2\sigma}.
\]
therefore $SAR(M)\ge 1+\frac{2\sigma}{n+n_u-2\sigma}$.
Through a similar argument, we infer that, when $n_r$ is odd, we have
\[
SAR(M)\ge 1+\frac{2\sigma-1}{n+n_u+1-2\sigma},
\]
which concludes the proof for the case in which $k$ is even.
We now consider the case in which $k$ is odd.
In this case, we have that $q_{\frac{k+1}{2}}=\frac{1}{2}$ and $q_{\frac{k-1}{2}}=\frac{1}{2}-\frac{1}{k}$.
As for the case in which $k$ is even, we assume that $F_\mu^{[-1]}(q_{\frac{k}{2}})=0$ and $F_\mu^{[-1]}(q_{\frac{k}{2}+1})=1$.
Moreover, $F_\mu^{[-1]}(q_i)\in(-\epsilon,0)$ if $i\le\frac{k}{2}$ and $F_\mu^{[-1]}(q_i)\in(1,1+\epsilon)$ for every other $i$,  where $\epsilon>0$ is a arbitrarily small parameter. 
In this case, we consider an instance in which the agents are some agents at $0$ and the remaining at $1$.
We denote with $l$ the set of agents at $0$ and with $r=n_r-l$ the number of agents at $1$.
To retrieve the optimal value of $l$, we use the same argument used in the proof of Theorem \ref{thm:LBmedian}.
Therefore we look for the value $l$ that maximizes the ratio between the optimal and the mechanism cost when the mechanism is restricted to place the facility at either $0$ or $1$.
Following the same argument used in the proof of Theorem \ref{thm:LBmedian}, we have that $l$ maximizes the ratio if
\[
    \frac{n}{\frac{n_u}{2}+n_r-l}=\frac{n}{l+(\frac{1}{2}-\frac{1}{k})n_u},
\]
hence $l=\frac{n_r-\frac{n_u}{k}}{2}$.

Let us now denote with $y\in[0,1]$ the position at which the truthful mechanism places the facility on this instance.
For every $y$, we use truthfulness to move either the agents at $0$ or the agents at $1$ to $y$ and define a sequence of probability measures $\mu_\ell$ that makes the ratio as large as possible.
To prove our lower bound, we then need to find the position $y$ that decreases the maximum ratio attainable after moving the agents and taking the limit.
Even though retrieving the optimal $y$ it is hard, we can still provide a lower bound on the SAR thanks to the following argument.
Given $y$, we denote with $\mathcal{L}(y)$ the maximum ratio attainable by moving all the agents at $0$ to $y$.
This function is decreasing in $y$.
Indeed, if we move all the agents from $0$ to $y$, the sequence of probability distributions we need to consider concentrates all the probability contained in $(0,1)$ at $1$, therefore
\begin{equation}
\label{eq:derivablefunction}
    \mathcal{L}(y)= 1 +\dfrac{\frac{3n_u}{k}(1-y)}{(\frac{1}{2}-\frac{1}{k})n_u+(\frac{n_r}{2}-\frac{n_u}{2k})(1-y)}.
    % \frac{n-y(n_r+\frac{2}{k}n_u)}{(1-y)l+(\frac{1}{2}-\frac{1}{k})n_u}-1= \frac{C}{A+B(1-y)}+\frac{2(1-y)\frac{n_u}{k}}{(\frac{1}{2}-\frac{1}{k})n_u+(1-y)(\frac{n_r}{2}-\frac{n_u}{2k})}
\end{equation}
% where $C$ is a constant value.
% 
By taking the derivative of \eqref{eq:derivablefunction}, obtain a function that is always negative, which allows us to conclude that $\mathcal{L}$ is non-increasing.

Similarly, we denote with $\mathcal{R}(y)$ the maximum ratio attainable by moving all the agents at $1$ to $y$.
By a similar argument, we have that $\mathcal{R}$ is increasing, thus the optimal $y$ is the point at which the graphs of $\mathcal{L}$ and $\mathcal{R}$ do intersect.
We can then retrieve a lower bound on the SAR of the truthful mechanism by fixing $y$ and by considering $\min\{\mathcal{L}(y),\mathcal{R}(y)\}$.
If we fix $y=\frac{1}{2}$, we have that, since $l\le n_r-l$, $\min\{\mathcal{L}(y),\mathcal{R}(y)\}=\mathcal{L}(y)$.
We therefore conclude that
\begin{equation}
    SAR(M)\ge 1+\dfrac{6\frac{n_u}{k}}{n+n_u-5\frac{n_u}{k}}=1+\dfrac{6\sigma}{n+n_u-5\sigma},
\end{equation}
which concludes the proof.
\end{proof}

\section{Extension to Two Facilities}

In this appendix, we extend our study to the case in which we have two facilities to place.

\subsection{Setting Statement}

We now consider the case in which we have two facilities with capacity $c$ to place, so that the total amount of agents the facilities can serve is $n=2c$.
We denote with $n_r\le n=2c$ the number of agents reporting their position to the mechanism.
Since we have two facilities with a capacity limit, the mechanism must elicit the positions of the facilities and then coordinate the agents by determining an agent-to-facility matching \cite{aziz2020facility}.
This ensures that no facility is overloaded with agents.
Given the agent report $\vec x$, let $y_1$ and $y_2$ the positions of the facilities and $\gamma$ be the agent-to-facility matching.
We denote with $n_u^{(1)}$ and $n_u^{(2)}$ the spare capacity of the facility at $y_1$ and the spare capacity of the facility at $y_2$, that is $n_u^{(j)}=c-\#\{i\in[n]\;\;\text{such that}\;\;(i,j)\in \gamma\}$.
To coordinate the aleatory agents, we need to determine a function $\fg:\erre\to\{y_1,y_2\}$ that maps the realization of $X\sim \mu$ into the two facilities.
Since $y_1$ has a spare capacity of $n_u^{(1)}$ and the facility at $y_2$ has a spare capacity of $n_u^{(2)}$, we define
\begin{equation}
    \label{eq:optimal_f}
    \fg:x\to
\begin{cases}
    y_1 \quad &\text{if}\quad x\le F_\mu^{[-1]}(\frac{n_u^{(1)}}{n_u})\\
    y_2 \;&\text{otherwise}
\end{cases}.    
\end{equation}

Notice that $\fg$ is well and uniquely defined by $\vec x$, $\vec y=(y_1,y_2)$, $\gamma$, and $\mu$.
Finally, given the agents' reports $\vec x$ we define the expected ex-ante social cost associated with the tuple $(\vec y,\gamma,\mu)$ as 
% it follows
% 
\begin{equation}
\label{eq:ESCttwofac}
    \ESC(\vec x; \vec y, \gamma, \mu)=\sum_{(i,j)\in \gamma}|x_i-y_j|\gamma_{i,j}+n_u\EE_{X\sim \mu}[|X-\fg(X)|].
\end{equation}
Given a vector $\vec x$ containing all the positions of the deterministic agents, the optimal $2$-FLPAA is the tuple $(y_1,y_2,\gamma)$ that minimizes $\ESC$.
We then denote with $\ESC_{opt}(\vec x;\mu)$ the minimum ESC attainable on instance $\vec x$, that is $\ESC_{opt}(\vec x;\mu)=\min_{\vec y,\gamma}\ESC(\vec x; \vec y, \gamma, \mu)$.
Similarly, we denote with $\ESC_{M}(\vec x;\mu)$ the ESC attained by the mechanism $M$, that is $\ESC_{M}(\vec x;\mu)=\ESC(\vec x; M(\vec x), \mu)$, where $M(\vec x)=(\vec y,\gamma)$.

\begin{theorem}
\label{thm:2fac_manyinfers}
    The function $f$ has the following properties:
\begin{enumerate*}[label=(\roman*)]
    \item \label{prop1} If we take $n_u$ samples of $X$, the expected number of agents that $f$ assigns to $y_1$ is $n_u^{(1)}$. Similarly, the expected number of agents that $f$ assigns to $y_2$ is $n_u^{(2)}$, that is $n_u\EE_{X\sim\mu}[\mathbf{Id}_{f(X)=y_j}]=n_u^{(j)}$ for $j=1,2$.
    Therefore the expected number of agents assigned to each facility is equal to the capacity of the facility.
    \item Amongst the functions $f$ that satisfy property \ref{prop1}, $f$ is the one that minimizes the expected ex-ante social cost of the allocation, that is $\EE[|X-f(X)|]\le \EE[|X-g(X)|]$ for every $g$ that satisfies  property \ref{prop1}.
    \item The function $f$ allows us to split $\mu$ into two measures with disjoint support, that is $\mu=F_{\mu}(z)\mu_{\le z} + (1-F_{\mu}(z))\mu_{>z}$, where $z=F_\mu^{[-1]}(\frac{n_u^{(1)}}{n_u})$.
    In particular, the expected ex-ante social cost can be written as 
    % $$\ESC(\vec x; \vec y, \gamma, \mu)=\sum_{(i,j)\in\gamma}|x_i-y_j|+n_u^{(1)}\EE_{X_1\sim\mu_{\le z}}[|X_1-y_1|]+n_u^{(2)}\EE_{X_2\sim\mu_{>z}}[|X_2-y_2|]$$.
    % \begin{equation}
    %     \ESC(\vec x; \vec y, \gamma, \mu)=\sum_{(i,j)\in\gamma}|x_i-y_j|+n_u^{(1)}\EE_{X_1\sim\mu_{\le z}}[|X_1-y_1|]+n_u^{(2)}\EE_{X_2\sim\mu_{>z}}[|X_2-y_2|].
    % \end{equation}
    % 
    % \]
\end{enumerate*}
\[
\ESC(\vec x; \vec y, \gamma, \mu)=\sum_{(i,j)\in\gamma}|x_i-y_j|+n_u^{(1)}\EE_{X_1\sim\mu_{\le z}}[|X_1-y_1|]+n_u^{(2)}\EE_{X_2\sim\mu_{>z}}[|X_2-y_2|].
\]
\end{theorem}

\begin{proof}
% [Proof of Theorem \ref{thm:2fac_manyinfers}]
We divide the proof into three points.
    \begin{enumerate}
        \item By definition $\fg(X)=y_1$ if and only if $X\le F_\mu^{[-1]}(\frac{n_u^{(1)}}{n_u})$, thus $P(\fg(X)=y_1)=\frac{n_u^{(1)}}{n_u}$.
        We then infer that $n_u\EE[\#\{X_i\,\text{s.t.}\;\fg(X_i)=y_1\}]=n_u^{(1)}$, which concludes the proof.
        \item Let $g$ be a function such that $\sum_{i=1}^{n_u}\EE[\#\{X_i\,\text{s.t.}\;g(X_i)=y_j\}]=n_u^{(j)}$ for $j\in[2]$.
        Then it must be that $g$ is a transportation map between $\mu$ and the probability measure $\nu=\frac{n_u^{(1)}}{n_u}\delta_{y_1}+\frac{n_u^{(2)}}{n_u}\delta_{y_2}$.
        % ù
        It is well known that the a map minimizing $\sum_{i=1}^{n_u}\EE[\#\{X_i\,\text{s.t.}\;g(X_i)=y_j\}]=n_u^{(j)}$ must be monotone (see \cite{villani2009optimal}) which concludes the proof of this point.
        \item First, notice that $\mu=F_\mu(z)\mu_{\le z}+(1-F_\mu(z))\mu_{>z}$ for every $z\in \erre$, thus the relation holds for $z=F_\mu^{[-1]}(\frac{n_u^{(1)}}{n_u})$, in which case $\mu=\frac{n_u^{(1)}}{n_u}\mu_{\le z}+(1-\frac{n_u^{(1)}}{n_u})\mu_{>z}$.
        To conclude the proof, we need to plug this relation into \eqref{eq:ESCttwofac}.
    \end{enumerate}
\end{proof}

Given a facility location $\vec y$ and an agent-to-facility matching $\gamma$, the cost of a reporting agent is $c_i(x_i,\vec y)=|x_i-y_j|$, where $(i,j)$ is the unique edge in $\gamma$ adjacent to agent $i$.
The ex-ante cost of an aleatory agent is $c_j(y_1,y_2)=\EE_{X\sim \mu}[|X-\fg(X)|]$, where $\fg$ is defined as in \eqref{eq:optimal_f}.
Before we move to the mechanism design aspects of the $2$-FLPAA, we characterize the optimal solution to the problem.
Like for the one facility case, the optimal solution to the $2$-FLPAA has a closed form.

\begin{theorem}
\label{thm:opt2facilities}
    Given two facilities with capacity $c$ and the agents' reports $\vec x$, let $F_{\lambda,\mu,\vec x}$ be defined as in \eqref{eq:flam}.
    If there exists $z\in\erre$ such that $z=F^{[-1]}_{\lambda,\mu,\vec x}(0.5)$, the optimal solution to the $2$-FLPAA is the tuple $(y_1,y_2,\gamma)$ where
    \begin{enumerate*}[label=(\roman*)]
        \item $y_1=F^{[-1]}_{\lambda,\mu,\vec x}(0.25)$, $y_2=F^{[-1]}_{\lambda,\mu,\vec x}(0.75)$; and
        \item $\gamma=\{(i,j)\}$ is defined as $(i,j)\in\gamma$ if and only if $x_i\le F^{[-1]}_{\lambda,\mu,\vec x}(0.5)$ and $j=1$ or $x_i>F^{[-1]}_{\lambda,\mu,\vec x}(0.5)$ and $j=2$.
    \end{enumerate*}
    If there is no point $z\in\erre$ such that $\Flam(z)=0.5$, $\gamma$ assigns every agent whose position is strictly lower than $z$ to $y_1$, the agents strictly higher than $z$ to $y_2$. 
    The agents at $z$ is assigned to either $y_1$ or $y_2$ depending on which allocation induces the lowest cost.
\end{theorem}

\begin{proof}
% [Proof of Theorem \ref{thm:opt2facilities}]
    We first prove the statement for the case in which there exists $z\in\erre$ such that $\Flam(z)=\frac{1}{2}$.
    The proof for the other case follows by a similar but more delicate argument.
    Given $\vec x$, let $(\vec y,\gamma)$ be the optimal facility location.
    First, we notice that $\gamma$ and $\fg$ split $\erre$ into two sets, namely $A_1$ and $A_2$, such that $A_1\cap A_2=\emptyset$ and $A_1\cup A_2=\erre$.
    The partition has the following property if $x_i\in A_1$ then $(i,1)\in\gamma$ and if $x\in A_1$ then $\fg(x)=y_1$, hence $A_1$ is the set of agents that are served by the facility at $y_1$.
    Similarly, $A_2$ is defined in such a way that if $x_i\in A_2$ then $(i,2)\in\gamma$ and if $x\in A_2$ then $\fg(x)=y_2$.
    By definition of $\fg$ and since $\gamma$ respects the capacity limits of the facilities, we have that $\#\{x_i\in A_1\}+n_u\mu(\fg^{(-1)}(A_1))=k$, thus we have $\frac{\#\{x_i\in A_1\}}{n}+(1-\lambda)\mu(\fg^{(-1)}(A_1))=\frac{1}{2}$.
    If $(\vec y,\gamma)$ is optimal, we must have that $A_1$ and $A_2$ are connected intervals.
    Since the capacity of the facilities are equalt, it must be that $A_1=(-\infty,\Flam^{[-1]}(\frac{1}{2})]$ and $A_2=(\Flam^{[-1]}(\frac{1}{2}),-\infty)$.
    Lastly, since $\vec y$ is optimal, it must be that $y_1$ minimizes the $\ESC$ of the agents in $A_1$, thus $y_1=\Flam^{[-1]}(0.25)$.
    By the same argument, we infer $y_2=\Flam^{[-1]}(0.75)$.
    Lastly, notice that by definition of $A_1$ and $A_2$, the optimal $\gamma$ is such that $(i,j)\in\gamma$ if and only if $x_i\le \Flam^{[-1]}(0.5)$ and $j=1$ or $x_i>\Flam^{[-1]}(0.5)$ and $j=2$, which concludes the proof.
    When there is no $z\in\erre$ such that $F_\mu(z)=\frac{1}{2}$, it must be the case that one or more agents are located at $y^*=\sup\{t\in\erre\;\text{s.t.}\;F_\mu(t)\le \frac{1}{2}\}=\inf\{t\in\erre\;\text{s.t.}\; F_\mu(t)\ge \frac{1}{2}\}$.
    In this case, we need to check all the possible cases on how to split the agents at $y^*$ and compute the associated ex-ante social cost.
    Once the matching $\gamma$ minimizing the ex-ante social cost is retrieved, we define the function $\fg$ is determined accordingly.
\end{proof}

As for the case in which we need to locate a single facility, given $\mu$, $n_r$, and $c$ the optimal position for the facilities belong to a discrete set.

\begin{corollary}
\label{corr:2FLPAA}
    The optimal locations $y_1$ and $y_2$ belong to a discrete set $y_1,y_2\in\{x_i\}_{i\in[n_r]}\cup\{F_\mu^{[-1]}(\frac{2j-1}{2n_u})\}_{j\in[n_u]}$.
\end{corollary}

\begin{proof}
% [Proof of Corollary \ref{corr:2FLPAA}]
    By Theorem \ref{thm:opt2facilities}, we have that $y_1=\Flam^{[-1]}(0.25)$ and $y_2=\Flam^{[-1]}(0.75)$.
    Let us consider $y_1$, the same argument can be applied to $y_2$. 
    If $y_1=x_j$ for some $j\in[n_r]$, then the proof is complete.
    If $y_1\neq x_j$ for every $j$ it means that $\Flam(y_1)=\frac{1}{4}$, hence
    \[
        \Flam(y_1)=\lambda\frac{\#\{x_i\le y_1\}}{n_r}+(1-\lambda)F_\mu(y_1)=\frac{\#\{x_i\le y_1\}}{2k}+(1-\lambda)F_\mu(y_1)=\frac{1}{4},
    \]
    where we used the fact that $\lambda=\frac{n_r}{n}$ and $n=2c$.
    We then have that
    \[
        \frac{n_u}{2c}F_\mu(y_1)=\frac{1}{4}-\frac{\#\{x_i\le y_1\}}{2c} \; \iff \; F_\mu(y_1)=\frac{c}{2n_u}-\frac{\#\{x_i\le y_1\}}{n_u}.
    \]
    We then conclude the proof by following the same argument used in the proof of Corollary \ref{cor:essentialpoints}.
\end{proof}

\subsection{Mechanism Design for the \texorpdfstring{$2$}{2}-FLPAA}

We now extend the results obtained for the one facility case to the $2$-FLPAA.
We divide the presentation into three sections depending how many quantiles the mechanism designer can query.

\subsubsection{\texorpdfstring{$n_u$}{k}-quantiles Case}

We first consider the case in which the mechanism designer has access to $n_u$ quantiles of $\mu$.
Owing to Corollary \ref{corr:2FLPAA}, we focus our attention to the case in which the quantiles to query are  the ones induced by $\vec q=(\frac{1}{2n_u},\frac{3}{2n_u},\dots,\frac{2n_u-1}{2n_u})$.
Unfortunately, in this case, the optimal mechanism is not truthful.

\begin{example}
    Let us consider $c=5$, so that the total capacity of the agents is $10$, let us set $n_r=8$.
    Let $\mu$ be the uniform distribution on the interval $[0,1]$.
    Let $\vec x=(0,1,1,2,9,9,9,9)\in\erre^8$ be the vector containing the agents' position.
    The optimal solution places $y_1$ at $0.75$, $y_2$ at $9$ and assigns the agents located at $0$ and $1$ to $y_1$, and all the others at $y_2$.
    The agent at $2$ can manipulate.
    Indeed, if it manipulates by reporting $0.75$ rather than $2$, the mechanism still places $y_1$ at $0.75$ and $y_2$ at $9$, however, in this case the manipulative agent is assigned to $y_1$, thus its cost decreases from $7$ to $1.25$.
\end{example}
When $n_r\le c$, there exists a truthful mechanism that places the facilities at the optimal positions.

\begin{mechanism}
    For every $c\in\enne$ and $n_r\le c$, let $\vec f=(F_\mu^{[-1]}(\frac{1}{2n_u}),\dots,F_\mu^{[-1]}(\frac{2n_u-1}{2n_u}))$.
    Given $\vec x$ in input the Pseudo Optimal Mechanism (POM) places the facilities at $z_{\floor{\frac{c+1}{2}}}$ and $z_{n-\floor{\frac{c}{2}}}$, where $\vec z$ is the vector obtained by reordering $(\vec x,\vec f)$.
    Then POM assigns every agent to the facility that is closer to the position they report. 
\end{mechanism}
Since $n_r\le c$ the POM never overloads a facility, thus is well defined.

\begin{theorem}
\label{thm:POMproperties}
    The POM is truthful. Moreover, we have that $\SAR(POM)=3$. 
\end{theorem}

\begin{proof}
% [Proof of Theorem \ref{thm:POMproperties}]
    % 
    First, we show that POM is truthful.
    Notice that since every agent is assigned to its closest facility, a manipulative agent must manipulate in such a way that at least one of the two facilities gets closer to their position.
    Therefore, it is easy to see that an agent whose real position is on the left of $y_1$ cannot manipulate in a way that makes $y_1$ get closer to its position, namely $x_i$.
    Indeed, if the agent at $x_i$ reports a position $x_i'\le x_i$, the output of the mechanism does not change, while if $x_i'>x_i$ then the position of the facilities move further to the right, hence the cost of the agent increases.
    In a similar way, we handle the case in which $x_i>y_2$ and $y_1<x_i<y_2$.
    We now compute the SAR of the POM.
    Owing to Theorem \ref{thm:opt2facilities}, we have that the mechanism places the facilities at the optimal positions, however the agent-to-facility matching $\gamma$ is sub-optimal.
    By the same argument used in the proof of Theorem \ref{thm:SAR0inf}, we have that the ratio between the mechanism cost and the optimal cost increases if we consider a sequence of measures that concentrate all the mass on the left of $y_1$ at $y_1$.
    Similarly, we increase the ratio by moving every agent on the left of $y_1$ to $y_1$.
    Likewise, the ratio between the cost of the mechanism and the optimal cost increases if we move all the agents and the mass that the measure locates to the right of $y_2$ to $y_2$.
    Notice that all these modifications do not alter the output of the mechanism nor the optimal position of the facilities.
    Let us now consider the agents whose position $x_i$ is between $y_1$ and $y_2$, that is $y_1<x_i<y_2$.
    Without loss of generality, we assume that all these agents are closer to $y_1$ than to $y_2$, so that the mechanism assigns them to $y_1$.
    If the optimal solution assigns them to $y_1$ as well, there is nothing to prove, as the optimal cost and the mechanism cost do coincide.
    To avoid this, we consider instances in which the probability measure concentrates all the probability between $y_1$ and $y_2$ close to $y_1$.
    Notice that, in this case, the ratio between the cost of the mechanism and the optimal cost increases as the agents are located closer to $\frac{y_1+y_2}{2}$.
    Moreover, the ratio increases as all the probability that $\mu$ assigns to the set $(y_1,y_2)$ concentrates at $y_1$.
    Taken a sequence of probability measures $\mu_\ell$ satisfying all these properties, we have that
    \[
        \lim_{\ell\to\infty}\ESC_{POM}(\vec x;\mu_\ell)=\ell \Big|\frac{y_1-y_2}{2}\Big|+\ell|y_1-y_2|
    \]
    while for the optimal cost we have that
    \[
        \lim_{\ell\to\infty}\ESC_{opt}(\vec x;\mu_\ell)=ell \Big|\frac{y_1-y_2}{2}\Big|,
    \]
    since both quantities do not depend on $\mu$, we retrieve that the SAR of the POM is
    \[
    SAR(POM)=\frac{\ell|y_1-y_2|+2\ell|y_1-y_2|}{\ell|y_1-y_2|}=3.
    \]
\end{proof}

Notice that when $n_r>c$, the POM is no longer well-defined.
For this reason, we need to introduce a different mechanism to handle this case.

\begin{mechanism}
    Let $n_r > c$ and let $\vec f$ be the the quantiles associated with $\vec q=(\frac{1}{2n_u},\dots,\frac{2n_u-1}{2n_u})$.
    Given $\vec x$, let $\vec z$ be the vector obtained by reordering the entries of $(\vec x,\vec f)$ increasingly.
    Then, the output of the Amended Quartiles Mechanism (AQM) on $\vec x$ is defined as follows
    \begin{enumerate*}[label=(\roman*)]
        \item we define $\vec y=(y_1,y_2)$ where $y_1=\max\{x_{n_r-k},z_{\ceil{\frac{c}{2}}}\}$ and $y_2=\min\{x_{k+1},z_{n-\floor{\frac{c}{2}}}\}$; and
        \item $\gamma$ as $(i,j)\in \gamma$ if and only if $|x_i-y_j|=\min\{|x_i-y_1|,|x_i-y_2|\}$; that is every agent is assigned to the facility that is closer to their report.
    \end{enumerate*}
\end{mechanism}

\begin{theorem}
\label{thm:AQMtheorem}
    The AQM is well-defined and truthful.
    Moreover, $\SAR(AQM)\le 3(c-1)$.
\end{theorem}

\begin{proof}
    % [Proof of Theorem \ref{thm:AQMtheorem}]
    First, we show that the mechanism is truthful.
    Toward a contradiction, let us assume that there exists a instance $\vec x$ in which the agent at $x_i$ can manipulate.
    Let $\vec y$ be the vector containing the positions of two facilities on the truthful input and $\gamma$ the associated matching.
    If $x_i\le y_1$, we notice that the agent cannot manipulate.
    Indeed, by definition of AQM, reporting a position $x_i'$ that is on the left of $x_i$ does not change the output of the mechanism, while reporting a position that is on the right of $x_i$ moves the two facilities to the right, which increases the cost.
    Similarly, we can show that no agent located at the right of $y_2$ can manipulate the mechanism.
    If we show that no agent located between $y_1$ and $y_2$ can manipulate, we conclude the proof.
    Since $x_i$ is assigned to its closest facility on the truthful input, the manipulating agent lowers its cost only if it is able to move one of the two facilities closer to them.
    We now show that the agent cannot misreport in such a way that $y_1$ becomes closer to $x_i$.
    Notice that if the agent at $x_i$ reports a position $x_i'\ge x_i$, the position of $y_1$ does not change.
    Through a similar argument, we handle the other case.
    Indeed, let $y_1'$ be the position of the leftmost facility returned by AQM on the input $(x_i',x_{-i})$ where $x_i'<x_i$ and $x_{-i}$ is the vector containing the reports of all the agents except the one located at $x_i$.
    Lastly, we compute the SAR of AQM.
    Let $\mu$ be a probability distribution.
    Owing to Theorem \ref{thm:opt2facilities} and to the definition of AQM, we have that the facility positions returned by the AQM, namely $\vec y=(y_1,y_2)$ are such that $y_1^*\le y_1\le y_2\le y_2^*$, where $y_1^*$ and $y_2^*$ are the optimal facility location.
    Owing to this property, we assume that no agents is such that $x_i<y_1^*$ or $x_i>y_2^*$.
    Likewise, we consider only probability measures that do not assign any probability to sets not included in $(y_1^*,y_2^*)$, as otherwise the resulting ratio between the mechanism cost and the optimal cost would be smaller.
    First notice that on every instance such that $y_1=y_1^*$ and $y_2=y_2^*$, we can repeat the argument used to prove Theorem \ref{thm:POMproperties} and show that the maximum ratio between the mechanism cost and the optimal cost is $3$.
\end{proof}

\subsubsection{Zero Information Case.}

When the mechanism designer cannot query any quantile, defining a truthful mechanism that has bounded SAR might be impossible depending on how many agents report their positions.
Indeed, if $n_r\le c$, no truthful mechanism can attain a bounded SAR.

\begin{theorem}
\label{thm:impoxresultfac2}
    Let $M$ be a truthful mechanism, if $n_r\le c$, then $\SAR(M)=+\infty$.
\end{theorem}

\begin{proof}
% [Proof of Theorem \ref{thm:impoxresultfac2}]
    Let $M$ be a truthful mechanism and let us consider the instance $\vec x=(1,1,\dots,1)\in\erre^{n_r}$.
    Let us denote with $\vec y=(y_1,y_2)=M(\vec x)$.
    Without loss of generality we can assume that $y_1,y_2\neq 0$.
    Let $\mu$ be the uniform distribution over the interval $[-1,1]$, so that $\rho_\mu(t)=\frac{1}{2}$ if $t\in[-1,1]$ and $\rho_\mu(t)=0$ otherwise.
    We define $\alpha_\ell$ as the probability distribution whose density is $\ell\rho_\mu(\ell x)$ and $\beta_\ell$ as the density whose probability distribution is $\ell\rho_\mu(\ell (x-1))$.
    We then define $\mu_\ell =\frac{k}{n_u}\alpha_\ell+\frac{n_u-k}{n_u}\delta_\ell$.
    By the same argument used in the proof of Theorem \ref{thm:SAR0inf}, we have that
    \[
        \lim_{\ell\to\infty}\ESC_{opt}(\vec x;\mu_\ell)=0,
    \]
    however, we have that $\lim_{\ell\to\infty}\ESC_{M}(\vec x;\mu_\ell)\ge |y_1|$, which allows us to conclude the proof.
\end{proof}

When $n_r>c$, the InnerGap Mechanism, introduced in \cite{ijcai2022p75}, is truthful, does not overload any facility, and attains bounded SAR.

\begin{mechanism}
    Given $\vec x\in\erre^{n_r}$ with $n_r>c$, the InnerGap Mechanism (IGM) returns $y_1=x_{n-c}$ and $y_2=x_{c+1}$.  
    Then the mechanism assigns every agents to the facility that is closer to the position they reported.
\end{mechanism}

\begin{theorem}
\label{thm:BEMtrutful+sar}
    The IGM is truthful and we have that $\SAR(IGM)\le 3(c-1)$.
\end{theorem}

\begin{proof}
% [Proof of Theorem \ref{thm:BEMtrutful+sar}]
    The truthfulness of the IGM has been shown in \cite{ijcai2022p75}.

    We now compute the SAR of the IGM.
    Given $\vec x$, let us denote with $y_1\le y_2$ the position at which the mechanism places the facilities
    We denote with $y_1^*\le y_2^*$ the optimal position of the facilities.
    Notice that by the same argument used in the proof of Theorem \ref{thm:POMproperties}, if $y_i^*=y_i$ for $i=1,2$, then the maximum ratio between the mechanism cost and the optimal cost is $3$.    
    We can then assume that at least one facility is located at a non optimal position.
    Let $z^*=\Flam^{[-1]}(0.5)$ and $z=\frac{y_1+y_2}{2}$.
    According to Theorem \ref{thm:opt2facilities}, every agent to the left of $z^*$ is assigned to $y_1^*$ in the optimal solution, while the rest are assigned to $y_2^*$. 
    Moreover, by the mechanism's definition, every agent left of $z$ is assigned to $y_1$, while others going to $y_2$.
    Without loss of generality, assume $z^*<z$.
    If no agents lie between $z^*$ and $z$, the optimal matching and the mechanism's matching do coincide.
    Thus, the worst-case scenario arises from independently studying the FLPAA induced by agents to the left and right of $z^*$.
    It is then easy to see that, in this scenario, the ratio between the mechanism's cost and the optimal cost is at most $c-1$.
    Consider now the case in which there is at least one agent $x_i$ in $[z^*,z]$. 
    The maximum number of agents in $[z^*,z]$ is at most $c-1$ by mechanism definition. 
    Here, the optimal matching differs from the mechanism's matching, allocating agents in $[z^*,z]$ to $y_1$, contrary to $y_2^*$ in the optimal case.
    The mechanism's cost is then maximized when all agents are at $z$, while the optimal cost decreases.
    Moreover, the ratio increases as the measure concentrates at $y_1^*$ and $y_2^*$ without shifting $z^*$.
    Denote $\Delta_i=|y_i-y_i^*|$ and $\Delta_y=\frac{|y_1-y_2|}{2}$. 
    It is easy to see that the optimal cost is larger than $\Delta_1+\Delta_y$, while the mechanism's cost is below $(3\Delta_y+2\Delta_1)(c-1)$.
    Finally, we have
    \[
        \frac{(3\Delta_y+2\Delta_1)(c-1)}{\Delta_y+\Delta_1}\le 3(c-1),
    \]
    concluding the proof.
\end{proof}

\subsubsection{\texorpdfstring{$k$}{k}-quantile case.}

We now study the case in which we have access to $k$ quantiles of $\mu$, where $k\in\{1,2,\dots,n_u-1\}$.
Notice that this framework includes the one quantile case.
First, we notice that depending on the quantiles we select, query, there might be no truthful mechanism that has bounded SAR.

\begin{example}
    Let us fix $c=3$, so that $n=6$.
    We fix $n_r=2$, thus $n_u=4$.
    Assume that we query for $3$ quantiles.
    If we query the quantiles associated with $q_1=0.05$, $q_2=0.1$, and $q_3=0.15$, any truthful mechanism has unbounded SAR.
    Indeed, let us consider the following instances indexed by $\ell\in\enne$.
    The report of the agents are $x_1=x_2=0$ for every $\ell$, while $\mu_\ell=\frac{\ell}{4}\mathcal{U}_{(-\frac{1}{\ell},0)}+\frac{3\ell}{4}\mathcal{U}_{(T-\frac{1}{\ell},T)}$, where $T$ is a parameter to fix.
    Given a truthful mechanism $M$, let $\vec y_\ell$ be the position at which it places the facilities on the instance.
    % ù
    If we set $T=y_1+y_2$, we have that the optimal cost of the instance converges to zero, while the cost of the mechanism is always larger than $0$.
\end{example}

When $\vec q=(\frac{1}{2k},\frac{3}{2k},\dots,\frac{2k-1}{2k})$, it is possible to define a truthful routine with bounded SAR.

\begin{mechanism}
    Let $\vec f=(F_\mu^{[-1]}(L(q)_1),\dots,F_\mu^{[-1]}(L(q)_{n_u}))$, where $\vec q=(\frac{1}{2k},\frac{3}{2k},\dots,\frac{2k-1}{2k})$. 
    For any $\vec x\in\erre^{n_r}$, let $\vec z\in\erre^{n}$ be the vector obtained by reordering $(\vec x,\vec f)\in\erre^n$.
    The Capped EndPoint Mechanism (CEM) is defined as follows.
    If $n_r\le c$, then 
    \[
    CEM(\vec x)=(z_1,z_n).
    \]
    If $n_r >c$, $CEM(\vec x)=(\max\{x_{n-k},z_{\ceil{\frac{c}{2}}}\},\min\{x_{k+1},z_{n-\floor{\frac{c}{2}}}\})$.
    In both cases, every agent is assigned to the facility closer to the position they reported.
\end{mechanism}

Since there are at most $k$ agents on the right of $y_1$ and $k$ agents on the left of $y_2$, CEM is well-defined and it does not overload any facility.

\begin{theorem}
\label{thm:2fackwhatever}
    The CEM is truthful.
    Moreover, we have $\SAR(CEM)\le 3(c-1)$.
\end{theorem}

\begin{proof}
% [Proof of Theorem \ref{thm:2fackwhatever}]
    First, we show that the CEM is truthful.
    When $n_r\le c$, it follows from the truthfulness of the EndPoint Mechanism.
    Let us consider the case $n_r>c$.
    Let $\vec x$ an instance on which the agent $x_i$ is able to manipulate by reporting $x_i'$ instead of its real position $x_i$.
    Let us denote with $y_1$ and $y_2$ the position at which the facilities are placed according to the truthful input and let $y_1'$ and $y_2'$ be the positions of the facilities when the agent at $x_i$ reports $x_i'$, i.e. when the mechanism is given in input $\vec x'=(x_i',\vec x_{-i})$.
    By definition of the mechanism, we have that $y_1\le y_2$ and $y_1'\le y_2'$.
    First, notice that since every agent is assigned to its closest facility, if $y_i=y_i'$ for $i=1,2$ the agents does not benefit from manipulating, therefore it must be that either $y_1\neq y_1'$ or $y_2\neq y_2'$.
    Second, we notice that if $x_i\le y_1$, then no agent can benefit by misreporting, since 
    \begin{enumerate*}[label=(\roman*)]
        \item if $x_i'\le x_i$ the output of the mechanism does not change; and 
        \item if $x_i\le x_i'$ then $y_i\le y_i'$, which increases the cost of the agent.
    \end{enumerate*}
    Similarly, no agent $x_i\ge y_2$ can manipulate.
    Lastly, a similar argument allows to handle the agents that are located in between $y_1$ and $y_2$.
    Indeed, the only way that an agent whose position is between $y_1$ and $y_2$ is to report a position that is outside $[y_1,y_2]$.
    If $x_i'\le y_1$, then the mechanism returns $y_1'<y_1$ and $y_2'=y_2$.
    Since the manipulative agent is assigned to $y_1'<y_1$, its cost increases, since $|y_1'-x_i|\ge |y_1-x_i|\ge \min\{|y_1-x_i|,|y_2-x_i|\}$.
    Finally, the bound on the SAR of the mechanism is obtained by the same arguments used to prove Theorem \ref{thm:BEMtrutful+sar} and \ref{thm:AQMtheorem}.
    Indeed, the bounds used in Theorem are applicable to this case and allows us to conclude that $SAR(CEM)\le 3(c-1)$.
\end{proof}

\end{document}